\pdfoutput=1

\documentclass[journal]{IEEEtran}
\makeatletter
\def\ps@headings{%
\def\@oddhead{\mbox{}\scriptsize\rightmark \hfil \thepage}%
\def\@evenhead{\scriptsize\thepage \hfil \leftmark\mbox{}}%
\def\@oddfoot{}%
\def\@evenfoot{}}
\makeatother
\IEEEaftertitletext{\vspace{-5mm}}

\usepackage{fancyhdr}

\usepackage{cite}
\ifCLASSINFOpdf
   \usepackage[pdftex]{graphicx}
\else
   \usepackage[dvips]{graphicx}
\fi
\usepackage[cmex10]{amsmath}

\usepackage{array}
\usepackage[textfont=small]{caption}
\usepackage[textfont=small]{subfig}

\usepackage{url}

\usepackage{amssymb}
\usepackage{verbatim}
\usepackage{color}

\begin{document}
\title{Energy-Efficient Strategies for Cooperative Multi-Channel MAC Protocols}
\author{Tie Luo,~\IEEEmembership{Member,~IEEE,}
        Mehul Motani,~\IEEEmembership{Member,~IEEE,}
        and~Vikram Srinivasan,~\IEEEmembership{Member,~IEEE}
}

\maketitle

\fancypagestyle{firststyle}
{
   \fancyhf{}
   \fancyhead[L]{IEEE TRANSACTIONS ON MOBILE COMPUTING, vol. 11, no. 4, pp. 553-566, April 2012}
   \fancyfoot[l]{A preliminary version of this work was presented at \cite{tie07mobicom}.}
}

\fancypagestyle{appendixpage}
{
   \rhead{\footnotesize\thepage}
   \fancyfoot[C]{Appendix is available at IEEE Xplore and \url{https://sites.google.com/site/luotie}.}
}

\renewcommand{\headrulewidth}{0pt}
\thispagestyle{firststyle}

\newcommand{\fref}[1]{Fig.~\ref{#1}}
\newcommand{\sref}[1]{Section~\ref{#1}}
\newcommand{\aref}[1]{appendix~\ref{#1}}
\newcommand{\tref}[1]{Table~\ref{#1}}
\newcommand{\pref}[1]{Prop.~\ref{#1}}
\newcommand{\mref}[1]{Theorem~\ref{#1}}
\newcommand{\dref}[1]{Def.~\ref{#1}}
\newtheorem{thm}{Theorem}
\newtheorem{defn}{Definition}
\newtheorem{prop}{Proposition}

\begin{abstract}
Distributed Information SHaring (DISH) is a new cooperative approach to designing multi-channel MAC protocols. It aids nodes in their decision making processes by compensating for their missing information via information sharing through other neighboring nodes. This approach was recently shown to significantly boost the throughput of multi-channel MAC protocols. However, a critical issue for ad hoc communication devices, i.e., energy efficiency, has yet to be addressed. In this paper, we address this issue by developing simple solutions which (1) reduce the energy consumption (2) without compromising the throughput performance, and meanwhile (3) maximize cost efficiency. We propose two energy-efficient strategies: {\em in-situ energy conscious DISH} which uses existing nodes only, and {\em altruistic DISH} which needs additional nodes called altruists. We compare five protocols with respect to the strategies and identify altruistic DISH to be the right choice in general: it (1) conserves 40-80\% of energy, (2) maintains the throughput advantage gained from the DISH approach, and (3) more than doubles the cost efficiency compared to protocols without applying the strategy. On the other hand, our study shows that in-situ energy conscious DISH is suitable only in certain limited scenarios.
\end{abstract}


\section{Introduction}\label{sec:intro-energy}
Using multiple channels in communication is key to improving the quality of service for wireless networks, and as a result, multi-channel MAC protocol design has attracted substantial attention from the research community. Tremendous effort has been made and various design approaches have been proposed, most of which require either multiple radios or time synchronization. Recently, \cite{tie09tmc} proposed a distinctive approach called DISH (Distributed Information SHaring), which uses a single radio but operates asynchronously. The authors designed a DISH-based protocol called CAM-MAC, in which neighboring nodes share control information with senders and receivers to compensate for their missed information in order to choose collision free channels or avoid busy receivers. Essentially, DISH can be viewed as a form of node cooperation, but there is a key difference: In traditional cooperation, intermediate nodes help relay data for source and destination nodes, which can be referred to as a {\em data-plane cooperation}. On the other hand, DISH requires nodes to send/receive control information only and thus can be referred to as a {\em control-plane cooperation}.

This approach has been extensively evaluated in \cite{tie09tmc} via the CAM-MAC protocol. The results demonstrate significant throughput improvement compared to protocols not using DISH, including existing representative multi-channel MAC protocols.

However, as DISH will be mainly used by ad hoc communication devices due to its distributed nature, energy efficiency becomes a crucial issue since those devices are mostly battery powered. The prior work \cite{tie09tmc} focused on throughput without considering energy consumption. In this paper, firstly to understand this issue particularly from a quantitative perspective, we carry out simulation to compare CAM-MAC with two protocols, Non-DISH and Non-DISH-psm where:
\begin{itemize}
\item Non-DISH is CAM-MAC with the element of DISH removed, i.e., neighbors do not share information with senders and receivers who will hence make decisions on their own. Basically, this is a (traditional) non-cooperative protocol.
\item Non-DISH-psm is Non-DISH with an ideal power saving mode (psm), where each node only turns on its radio when sending/receiving packets addressed from/to itself.
\end{itemize}
More protocol details will be described in \sref{sec:proto-energy}. The simulation results show that, although the throughput of CAM-MAC is 2.65 times Non-DISH and even more than Non-DISH-psm, its energy consumption is 2.94 times Non-DISH-psm and comparable to Non-DISH (detailed results will be given in \sref{sec:perf}). This conveys that there is potentially large space for improvement in energy efficiency for DISH.

To address this issue, we propose two energy-efficient strategies, {\em in-situ energy conscious DISH} and {\em altruistic DISH}, in this \ifdefined\thesis study. \else paper. \fi In the in-situ strategy, existing nodes rotate the responsibility of information sharing such that nodes without this responsibility can sleep when idle in order to save power. In the altruistic strategy, additional nodes called {\em altruists} are deployed to {\em take over} the responsibility of information sharing so that all the existing nodes can sleep when idle.

We conduct both qualitative and quantitative work to investigate the strategies with the following objectives: (1) reduce the energy consumption, (2) maintain or not compromise the high throughput achieved via DISH, and (3) maximize cost efficiency. Yet, the solution must be kept as simple as possible. By comparing five protocols with respect to the strategies, our study recommends altruistic DISH in general and in-situ energy conscious DISH only in certain limited scenarios.

We have also built a hardware test-bed and conducted experiments where the results have further confirmed our findings. Moreover, neither of the two strategies that we propose requires multiple radios nor time synchronization, which basically translates to lower cost, smaller hardware size and/or low complexity.

The rest of the paper is organized as follows. \sref{sec:preliminary-alt} explains DISH in more detail, and \sref{sec:qualitative} elaborates and gives a qualitative analysis of the strategies, where three important issues are identified to be addressed. These issues, optimal node deployment, cost efficiency, and throughput-energy tradeoff, are subsequently investigated in \sref{sec:deploy}, \sref{sec:costeff}, and \sref{sec:perf}, respectively. Then we discuss relevant issues in \sref{sec:discuss} and review related work in \sref{sec:relwork}. Finally, \sref{sec:conclud} concludes the paper.

\section{Understanding DISH}\label{sec:preliminary-alt}

Control information is crucial to communications but can be missing due to various reasons such as shadowing and noise. The dominant reason, however, in a multi-channel environment, is that nodes fail to tune radios to  certain channels in time, or that a radio can only listen to one channel at a time. This causes the multi-channel coordination (MCC) problem which has two variants: (1) channel conflict problem, created when a node selects a busy channel (being used by other nodes), and (2) deaf terminal problem, created when a sender attempts to communicate with a receiver that is on a different channel.

One category of solutions are to dedicate an extra radio to each channel or a common control channel in order not to miss information, as proposed by \cite{dca00,nas99,nas00,jain01,mup04,mah06}. However, such solutions will inevitably increase hardware cost and size (and energy consumption as well). Another category of solutions do not require multiple radios but require communication to be set up in specified time slots~\cite{chen03,mmac04,tmmac07} or require periodic channel switching according to certain sequences~\cite{chma00,chat00,ssch04}. Thus they rely on time synchronization which adds considerable complexity~\cite{so06sync} and degrades scalability~\cite{scale02mobihoc}, especially for multi-hop networks.

The basic idea of DISH is to compensate for nodes' missing information via cooperation. It exploits neighboring nodes as a resource to ``retrieve'' missing information from, like from a distributed database, when needed. The need for multiple radios or time synchronization, naturally becomes not necessary.

\subsection{DISH-p: A DISH-based Protocol}\label{sec:dish-detail}

For a more tangible understanding, we describe a DISH-based protocol called DISH-p (which was CAM-MAC described in \cite{tie09tmc}). In DISH-p, a sender and a receiver set up communication using PRA/PRB packets and then confirm using CFA/CFB packets. A neighbor will send INV packet if it identifies an MCC problem via the information conveyed by PRA/PRB.

There is one control channel and multiple data channels. On the control channel, a sender and a receiver exchange PRA/PRB (see \fref{fig:hsk-ctrl}) to select a data channel, and then exchange CFA/CFB to confirm the channel selection. The frame format is shown in \fref{fig:format}. If a neighbor identifies an MCC problem (via PRA or PRB), it will prepare to send an INV packet, during a cooperation collision avoidance period (CCAP), to alarm the sender or the receiver to back off. If there is no MCC problem identified by any neighbor (no INV will be sent), the sender and the receiver will switch to their chosen data channel and start DATA/ACK exchange. During DIFS and CCAP, carrier sensing is turned on to mitigate collisions via CSMA.

\begin{figure}[t]
\centering
\subfloat[Control channel handshake.]
{\includegraphics[width=.95\linewidth]{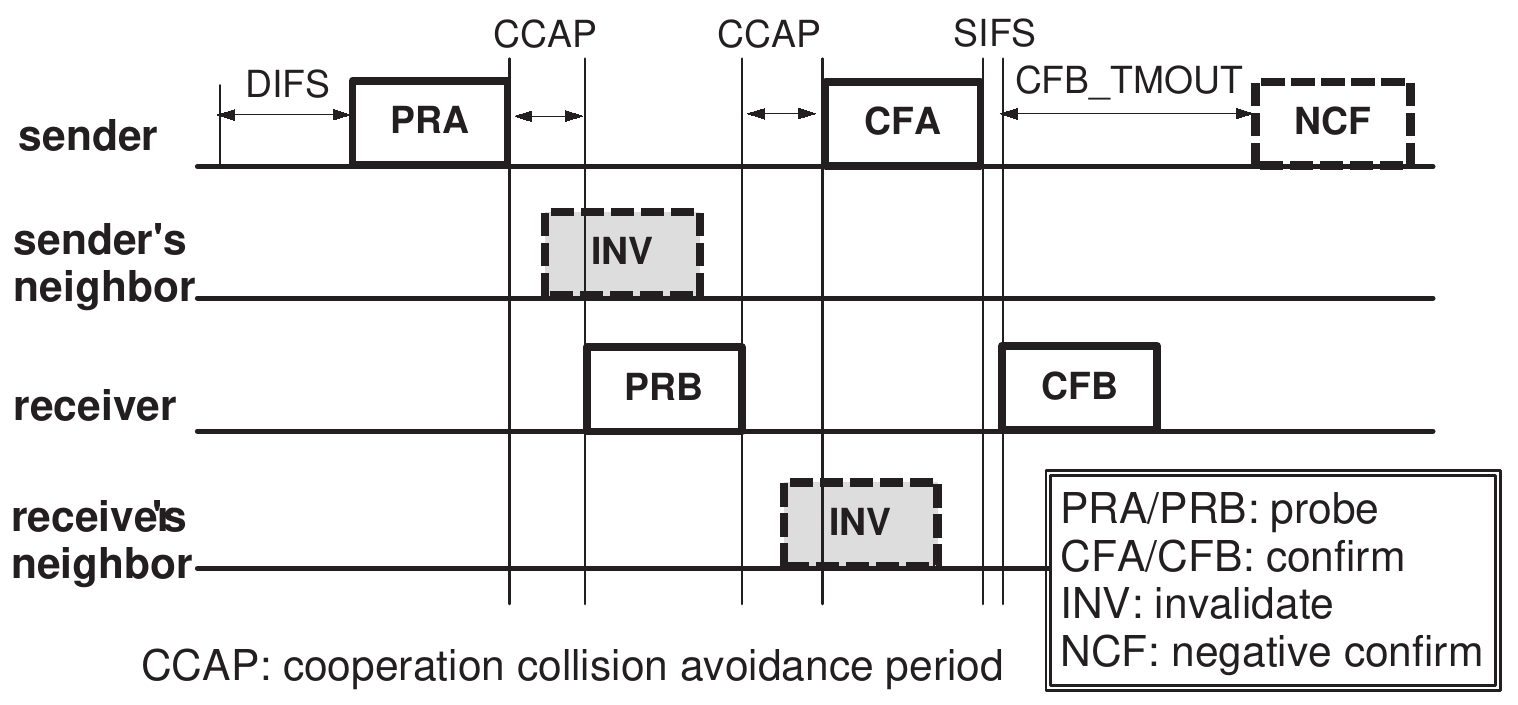}\label{fig:hsk-ctrl}}\vfil
\subfloat[Frame format. INV carries the channel usage information of an established and ongoing data exchange on a data channel (which engages the ``deaf'' receiver in the case of deaf terminal problem).]
{\makebox[\linewidth]{\includegraphics[width=.62\linewidth]{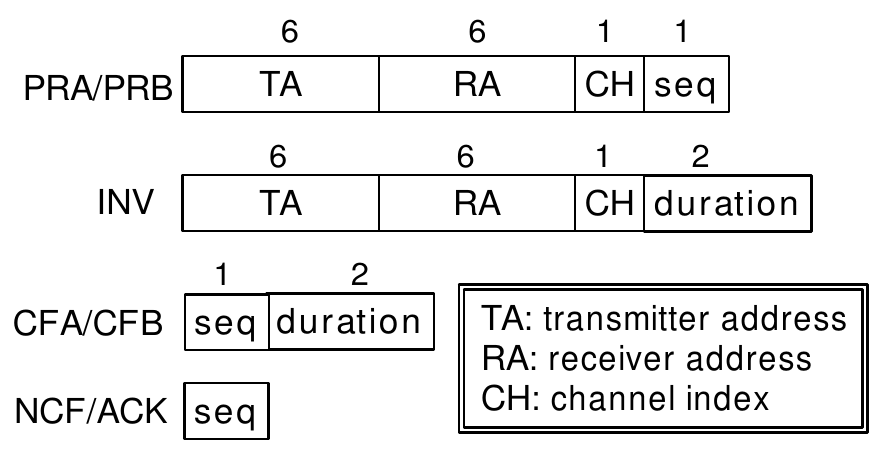}\label{fig:format}}}\vfil
\subfloat[Channel usage table. Each node maintains one to cache its overheard control information.]
{\makebox[\linewidth]{\includegraphics[width=.42\linewidth]{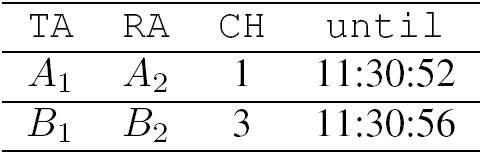}\label{fig:chtab}}}
\caption{Elements of the DISH-p protocol.}\label{fig:dishp}
\end{figure}

CCAP is introduced to mitigate the collision of multiple simultaneously sent INVs. A neighbor who identifies an MCC problem wil send INV only if it senses the control channel to be free for a period of \textsl{Uniform}[0, CCAP]. Hence a neighbor who sends INV will suppress its neighbors via CSMA.\footnote{CSMA does not avoid all collisions because not all the neighbors may hear each other. However, a collision of such still conveys an alarm to the sender/receiver because INV represents a {\em negative} message, and hence the sender/receiver will still back off. What is only compromised is that the sender/receiver will not know precisely how long at least it should back off and hence will have to estimate a backoff period, which is not a serious problem.} NCF is sent when the sender waits for CFB until timeout (due to the receiver receiving INV), in order to inform the sender's neighbors to disregard CFA.

The applicable scenarios of the protocol are mesh networks and ad hoc networks, not sensor networks. In sensor networks, data packets are usually small and the overhead of the control channel handshake will be significant. Even using a packet train would not suit because sensing traffic is usually periodic and not bursty.

\section{Energy-Efficient Strategies}\label{sec:qualitative}

The main challenge to achieving energy efficiency for DISH is that a prerequisite of information sharing is {\em information gathering}, a process that requires nodes to stay awake for overhearing, which presents a challenge for nodes to switch off radio when idle. The strategies we elaborate below meet this challenge and we also provide a qualitative analysis below.

\subsection{In-Situ Energy Conscious DISH}

In this strategy, all the existing nodes rotate the responsibility of information sharing (i.e., cooperation) such that nodes without the responsibility can sleep when idle.\footnote{We say that a node is idle if it is not engaged in sending/receiving its {\em own} packets. For example, overhearing (other packets) and waiting for free data channels (though with data packets in queue) are both idle.} There are two methods to implement this strategy:
\begin{itemize}
\item Probabilistic method: Each node decides whether to cooperate or not according to a (static or dynamic) probability. This is similar to probabilistic flooding \cite{prob_flood99mobicom,prob_flood05IWWAN,dyn05jpdc} and probabilistic routing \cite{wsna03prob_rt,tanya05pgr} in ad hoc networks,
and cluster-head rotating algorithms (e.g., LEACH~\cite{leach00} and HEED~\cite{heed04}) in sensor networks.
\item Voting method: nodes periodically vote or elect a subset of nodes to cooperate. This is similar to GAF\cite{gaf01mobicom}, Span\cite{span01mobicom}, PANEL\cite{panel07mass} and VCA\cite{vca07}.
\end{itemize}

An apparent advantage of the in-situ strategy is that it does not require additional nodes. On the other hand, a runtime probabilistic or voting mechanism must be introduced and must be (1) distributed, (2) fair (in terms of energy consumption), and (3) adaptive (to network dynamics such as traffic and energy drainage). These would introduce considerable complexity and overhead. In addition, it has to consider other factors as listed below.

First, the mechanism would rely on message broadcast as also used by \cite{prob_flood99mobicom,prob_flood05IWWAN,dyn05jpdc,gaf01mobicom,span01mobicom,vca07,tanya05pgr}.
However, broadcasting in a multi-channel environment is shown by \cite{so06sync} to be very unreliable and difficult because each broadcast can reach only a subset of neighboring nodes. Alternatively, broadcasts might be reduced or avoided by determining cooperative nodes based on geographic information, like in \cite{wsna03prob_rt,gaf01mobicom,panel07mass}.
However, this requires expensive GPS support or a distributed localization algorithm (e.g., \cite{bruck05mobihoc,caruso05infocom}) which introduces additional overhead and complexity to those incurred by rotation itself.

Second, rotating the responsibility of cooperation also involves other resource-consuming factors including two-hop neighbor discovery (shown in \cite{tie06cam,tie09tmc}) and the assessment of dynamic information (such as energy and traffic, like in \cite{leach00,span01mobicom,tanya05pgr}).

Third, how to integrate a probabilistic or voting mechanism into a legacy DISH protocol is a non-trivial problem and a viable solution is yet to be found.

In summary, the complexity, overhead, and unreliability of in-situ energy conscious DISH would consume considerable resource and eventually negate its possible performance gain. Nonetheless, for a quantitative understanding, we still evaluate this strategy using a {\sl Genie In-Situ} protocol (detailed in \sref{sec:proto-energy}) which establishes an upper bound for all such in-situ protocols.

\subsection{Altruistic DISH}

In this strategy, additional nodes called {\em altruists} are deployed to take over the responsibility of information sharing (i.e., cooperation) from the existing nodes, which we call {\em peers} to distinguish from altruists, so that peers can sleep when idle. Altruists are the same as peers in terms of hardware, but are different in terms of software: they solely cooperate (do not carry data traffic) and always stay awake.

An apparent drawback of this strategy is that it requires additional nodes. However, this is offset by substantive advantages. First, it is very simple to implement the strategy: one only needs to introduce a boolean flag to disable data related functions on altruists and cooperation related functions on peers. We have done this in both our simulation code and hardware implementation code. Equally importantly, there is no additional runtime mechanism and hence runtime overhead.

Second, unlike the in-situ strategy, this strategy does not have the multi-channel broadcasting problem. Altruists always stay on the same channel (control channel) and send/receive packets only on the control channel.

Third, this strategy is robust to network dynamics (such as traffic and residual energy). Every altruist is cooperative and will react to every MCC problem that it identifies; they do not need to adjust any parameter on the fly. In fact, even the deployment of altruists, which is an offline process, can be done with a constant number for any given peer density, as will be shown in \sref{sec:deploy}.

Fourth, since peers only carry data traffic and need not to cooperate, they are like nodes in traditional (non-DISH) networks and thus can adopt a legacy sleep-wake scheduling algorithm, where a lot of choices are available and will be provided in \sref{sec:relwork}.

Finally, unlike the in-situ strategy and the original DISH where cooperation is provided in an {\em opportunistic} manner---meaning that cooperative nodes are not always available, altruistic DISH provides cooperation in a {\em guaranteed} manner.

\subsection{Protocols to Investigate}\label{sec:proto-energy}

In the sequel, we investigate  {\sl Genie In-Situ} and {\sl Altruistic}, which are two protocols made by applying the above two strategies to DISH-p (original DISH protocol) respectively. For the purpose of comparison, we also need to introduce two non-DISH protocols, one with and the other without power saving, viz., {\sl Non-DISH} and {\sl Non-DISH-psm}. The following describes all the five protocols.

\begin{enumerate}
\item {\sl DISH-p}: the protocol described in \sref{sec:dish-detail}.
\item {\sl Non-DISH}: a (traditional) non-cooperative protocol, derived from DISH-p by by removing the cooperative element, i.e., neighbors do not  share control information with senders/receivers.
\item {\sl Non-DISH-psm}: Non-DISH with a power saving mode (PSM), where each node only turns on its radio when sending/receiving packets addressed from/to itself (i.e., they do not overhear). This is an ideal mode because it assumes a receiver can automatically wake up upon a communication request from a sender. We use this rather than adopt an existing sleep-wake scheduling algorithm (which will be reviewed in \sref{sec:relwork}) in order to avoid coupling performance to a specific algorithm. Besides, this still keeps our comparison fair because the same PSM will be used by all the other power-saving protocols ({\sl Genie In-Situ} and {\sl Altruistic}).
\item {\sl Genie In-Situ}: this protocol is DISH-p with the in-situ strategy applied. It uses a genie-aided (optimal) rotating mechanism in order to establish upper bound performance for the in-situ strategy. In this protocol, upon each occurrence of an MCC problem, the best neighbor will be chosen (by the genie) to cooperate\footnote{The best neighbor is a neighbor with the most helpful information when an MCC problem occurs. E.g. in a channel conflict problem where a node $u$ chooses a busy data channel which is used by multiple sender-receiver pairs (consider a multi-hop environment), the best neighbor is the one who knows which pair has the {\em longest residual time} in using that channel---this neighbor can inform node $u$ of the minimum duration to back off for.}  and all the other neighbors are treated as virtually sleeping (not consuming energy though having gathered information via overhearing) as per the ideal PSM.
\item {\sl Altruistic}: this protocol is DISH-p with the altruistic strategy applied. Altruists stay awake to gather information and, upon identifying an MCC problem, share information (cooperate). All existing nodes do not cooperate and they adopt the ideal PSM to sleep when idle.
\end{enumerate}

\subsection{Issues to Investigate}

There are three relevant issues that need to be addressed:
\begin{enumerate}
\item {\em Node deployment} (addressed in \sref{sec:deploy}): How to deploy altruists for Altruistic DISH.
\item {\em Cost efficiency} (addressed in \sref{sec:costeff}): We propose a metric called bit-meter-price (BMP) ratio which takes into account various factors to measure the overall performance of a protocol.
\item {\em Throughput-energy trade-off} (addressed in \sref{sec:perf}): Zooms in to specifically inspect the throughput and energy performance.
\end{enumerate}

In the rest of the study we assume an ad hoc network with static topology. Each node has a single half-duplex radio that can dynamically switch among all available channels but can only use one at a time. One channel is designated as a control channel and the others as data channels. Data channel selection is random, meaning that a sender/receiver randomly selects one from a list of data channels that it deems free based on its knowledge which it dynamically updates (e.g., channel usage table as in \fref{fig:chtab}).\footnote{Another channel selection method is first to try the previously used channel and, if not available, then random selection. Both methods were studied in \cite{tie09tmc} and shown, in most cases, to result in only marginal difference in the context of DISH.} Finally, we assume all links are bidirectional, i.e., if node $u$ can hear node $v$ then $v$ can hear $u$ as well.

\section{Optimal Node Deployment}\label{sec:deploy}

As a prerequisite, we need to develop a concept called {\em cooperation coverage}.

\subsection{Cooperation Coverage}\label{sec:cov}

\begin{defn}[UP and CUP]\label{def:up}
An unsafe pair (UP) is a pair of peers that can create MCC problems to each other. A covered unsafe pair (CUP) is an UP that both peers are within the transmission range of at least one common altruist.
\end{defn}

\begin{figure}[ht]
\centering
\includegraphics[width=.7\linewidth]{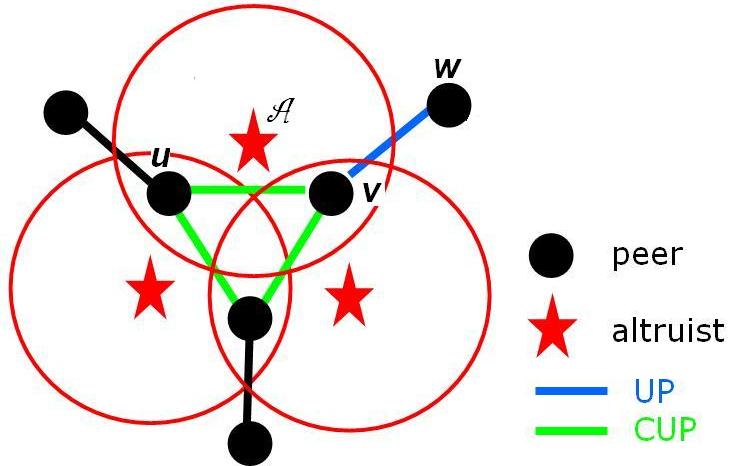}
\caption{Illustration of UP and CUP. Node pair $(u,v)$ is a CUP (covered by altruist $\mathcal{A}$) while $(v,w)$ is an UP. Each circle denotes the transmission range of an altruist.}\label{fig:up-cup}
\end{figure}

An illustration of UP and CUP is given in \fref{fig:up-cup}, and the necessary and sufficient condition for creating MCC problems (i.e., forming an UP) is given in \pref{prop:form_up}. Briefly speaking, two adjacent peers can create MCC problems if each of them has other communicable neighbor(s), because one peer may switch to a data channel and miss information of the other peer.

\begin{prop}\label{prop:form_up}
In an undirected graph where each vertex represents a peer and each edge represents the relationship between two neighboring peers, denote by $d_i$ the degree of an arbitrary vertex $i$. If PSM is not used, two adjacent vertices $i$ and $j$ form an UP if and only if:
\begin{enumerate}
\renewcommand{\labelenumi}{(\alph{enumi})}
\item $d_i\ge 2$, $d_j\ge 2$, and $d_i=d_j=2$ does not hold, or
\item $d_i=d_j=2$, and $i$ and $j$ are not on the same three-cycle (i.e., triangle).
\end{enumerate}
If PSM is used (peers sleep when idle), the above condition remains unchanged for the channel conflict problem, but changes to the following for the deaf terminal problem:

$d_i\ge 1$, $d_j\ge 1$, and $d_i=d_j=1$ does not hold.
\end{prop}

\begin{IEEEproof}See Appendix.\end{IEEEproof}
\thispagestyle{appendixpage}

\begin{defn}[Cooperation Coverage --- $p_{cov}$]\label{def:cov}
\[ p_{cov} \triangleq \frac{N_{cup}}{N_{up}}, \]
where $N_{cup}$ is the number of CUPs and $N_{up}$ is the number of UPs in a network.
\end{defn}
We say that a network achieves {\em full cooperation coverage} if $p_{cov} = 100\%$.

\begin{prop}\label{prop:coll_free}
Consider a network using altruistic DISH. In order to achieve free of MCC problems, full cooperation coverage is
\begin{enumerate}
\item necessary for a multi-hop network, and
\item necessary and sufficient for a single-hop network.
\end{enumerate}
\end{prop}

\begin{IEEEproof}See Appendix.\end{IEEEproof}

\subsection{Random Deployment}\label{sec:deploy_rand}

In random deployment, all nodes are uniformly distributed in a plane region.

\begin{thm}\label{thm:deploy_rand}
Consider an infinite network where peers and altruists are randomly distributed (as per a two-dimensional Poisson point process). If the peer density is $\rho_{peer}$, then in order to achieve a cooperation coverage of $p_{cov}$, the altruist density, $\rho_{alt}$, must satisfy
\begin{align}\label{eq:rhoalt}
    \rho_{alt} > - \frac{\ln (1-p_{cov}) }{(\frac{2\pi}{3} - \frac{\sqrt{3}}{2}) r^2}.
\end{align}
\end{thm}

\begin{proof}
Denote by $p_{ij}^{cov}$ the probability that an arbitrary UP $(i,j)$ is covered (i.e., is a CUP). By Definition~\ref{def:up}, $p_{ij}^{cov}$ is equivalent to the probability that at least one altruist exists in the common transmission range of $i$ and $j$, which is given by
\begin{align}\label{eq:p_pov}
    p_{ij}^{cov} = 1 - e^{-\rho_{alt} A_{ij}},
\end{align}
where $A_{ij}$ is the intersected area of $i$ and $j$'s transmission ranges, and can be proven using simple geometric techniques to be
\begin{align}\label{eq:intersect}
    A_{ij} = 2 r^2 \theta - r^2 \sin 2\theta,
\end{align}
where $\theta = \arccos \frac{d}{2r}$, $d$ is the Euclidean distance between $i$ and $j$, and $r$ is the transmission range.

To achieve $p_{cov}$ is equivalent to achieving $p_{ij}^{cov}>p_{cov}$ for all UPs $(i,j)$, meaning
\begin{align}\label{eq:ineq}
    \min_{(i,j)} p_{ij}^{cov} > p_{cov}.
\end{align}
According to \eqref{eq:p_pov}, $p_{ij}^{cov}$ is a monotonically increasing function of $A_{ij}$, and hence is minimized by minimizing $A_{ij}$. To minimize $A_{ij}$, consider the minimization domain, namely all UPs. According to \eqref{eq:intersect}, $A_{ij}$ is a monotonically decreasing function of $d$. Since $d\in[0,r]$, $A_{ij}$ is therefore minimized at $d=r$:
\begin{align*}
    \min_{(i,j)} A_{ij} = A_{ij}|_{d=r} = (\frac{2\pi}{3} - \frac{\sqrt{3}}{2}) r^2,
\end{align*}
and thus \eqref{eq:ineq} resolves to
\begin{align}
    \min_{(i,j)} p_{ij}^{cov} &= 1 - \exp(-\rho_{alt}\cdot \min_{(i,j)} A_{ij})\notag\\
        &= 1- \exp[- \rho_{alt}\cdot (\frac{2\pi}{3} - \frac{\sqrt{3}}{2}) r^2 ]\notag\\
        &> p_{cov}, \nonumber
\end{align}
which is then reduced to \eqref{eq:rhoalt}.
\end{proof}

\mref{thm:deploy_rand} gives the relationship between altruist density $\rho_{alt}$ and cooperation coverage $p_{cov}$. Note that $\rho_{peer}$ does not appear in \eqref{eq:rhoalt}. This is important because it implies that altruist deployment is independent of peer density and hence is remarkably simplified. This also makes significant practical sense because, in reality, the number of peers often varies or is uncertain.

\mref{thm:deploy_rand} also shows that $\rho_{alt}\rightarrow \infty$ if $p_{cov} = 100\%$. This tells network planers not to aim at full cooperation coverage in multi-hop networks. For single-hop networks, it is easy to see that a single altruist achieves full cooperation coverage.

\begin{figure}[ht]
\ifdefined\thesis
\centering\includegraphics[width=0.45\linewidth]{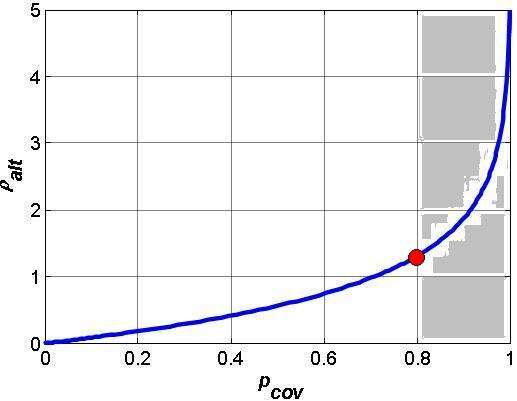}
\else
\centering\includegraphics[width=0.55\linewidth]{alt_vs_pcov.jpg}
\fi
\caption{$\rho_{alt}$ versus $p_{cov}$ (\mref{thm:deploy_rand}).}
\label{fig:alt_vs_pcov}
\end{figure}
\begin{table}[ht]
\renewcommand{\arraystretch}{1.15}
\centering
\caption{Some Discrete Values of $\rho_{alt}$ versus $p_{cov}$}
\label{tab:altruist}
\begin{tabular}{|c|c|c|c|c|c|c|c|}
    \hline
$p_{cov}$ & 50\% & 60\% & 70\% & 80\% & 90\% & 95\% & 99\% \\
    \hline
$\rho_{alt} >$ & 0.56 & 0.75 & 0.98 & 1.31 & 1.87 & 2.44 & 3.75 \\
    \hline
\end{tabular}
\end{table}
\fref{fig:alt_vs_pcov} plots the relationship between $\rho_{alt}$ and $p_{cov}$, and \tref{tab:altruist} enumerates some discrete values. We can see that beyond the point $(p_{cov} = 80\%$, $\rho_{alt}=1.31)$, $\rho_{alt}$ sharply increases, which indicates a cost spike. This motivates us to investigate the performance trend of altruistic DISH in a range including this point.

{\sl 1) Simulation Setup and Results}\label{sec:setup}

In our simulation, the metrics are aggregate end-to-end throughput and aggregate power consumption (including both peers and altruists if any). In order to compute power consumption, we conducted a survey of power consumption rates of commercial wireless cards. According to \cite{cisco-wifi}, a Cisco Aironet 350 series WiFi card consumes 2250/1350/75 mW in TX/RX/SLEEP state. According to \cite{infocom01invest} (with some simple calculation), an IEEE 802.11 WaveLAN PC card consumes 1327/967/843/66 mW in TX/RX/IDLE/SLEEP state in the 2Mbps category, and 1346/901/741/48 mW in the 11Mbps category. According to other respective sources, Intel Pro 2011, 3Com xJack, Compaq WL1000 and Siemens SS1021 all have the rates with the similar ratio as the above.\footnote{A (relative) ratio matters more than (absolute) rates as this is a comparative study which focuses on the difference between protocols.} Therefore, we use the average rates based on our survey, namely 25/18/15/1 $\times$50mW in the TX/RX/IDLE/SLEEP states respectively, to calculate the power consumption in simulation.

We set up the simulation as follows. Nodes are randomly placed in a plane area of 100m$\times 100$m for single-hop networks and 1500m$\times 1500$m for multi-hop networks. The radio transmission range is 250m and the interference range is 500m. The capture effect is enabled with a threshold of 6dB. In single-hop networks, $n$ peers randomly form $n/2$ disjoint flows. In multi-hop networks, $n$ peers randomly form $n$ non-disjoint flows (each peer is the source of one flow and also the destination of another flow). Shortest path routing is used. There are one control channel and five data channels with bandwidth 1Mbps each. Packet arrival is Poisson, and data payload is 2KB. PLCP is 15 bytes (header 6 bytes and preamble (short) 9 bytes).
SIFS is 10$\mu$s and CCAP is 35$\mu$s.\footnote{CCAP needs to be sufficient for a node to detect signal transmission, not need to receive a complete message.} Channel switching delay is ignored because it is common to all the protocols in this comparative study, and is not long (80$\mu$s according to \cite{ssch04}, equivalent to transmitting 10 bytes on an 1Mbps channel).
We use a discrete-event simulator which we developed on Fedora Core 5 with a Linux kernel of version 2.6.9. Each simulation is terminated after a total of 100,000 data packets are sent. All results are averaged over 15 randomly generated networks.

\begin{figure}[t]
\centering
\ifdefined\thesis
    \subfloat[Aggregate throughput.]{\includegraphics[trim=4mm 0 1.1cm 6mm, clip, width=0.48\linewidth]{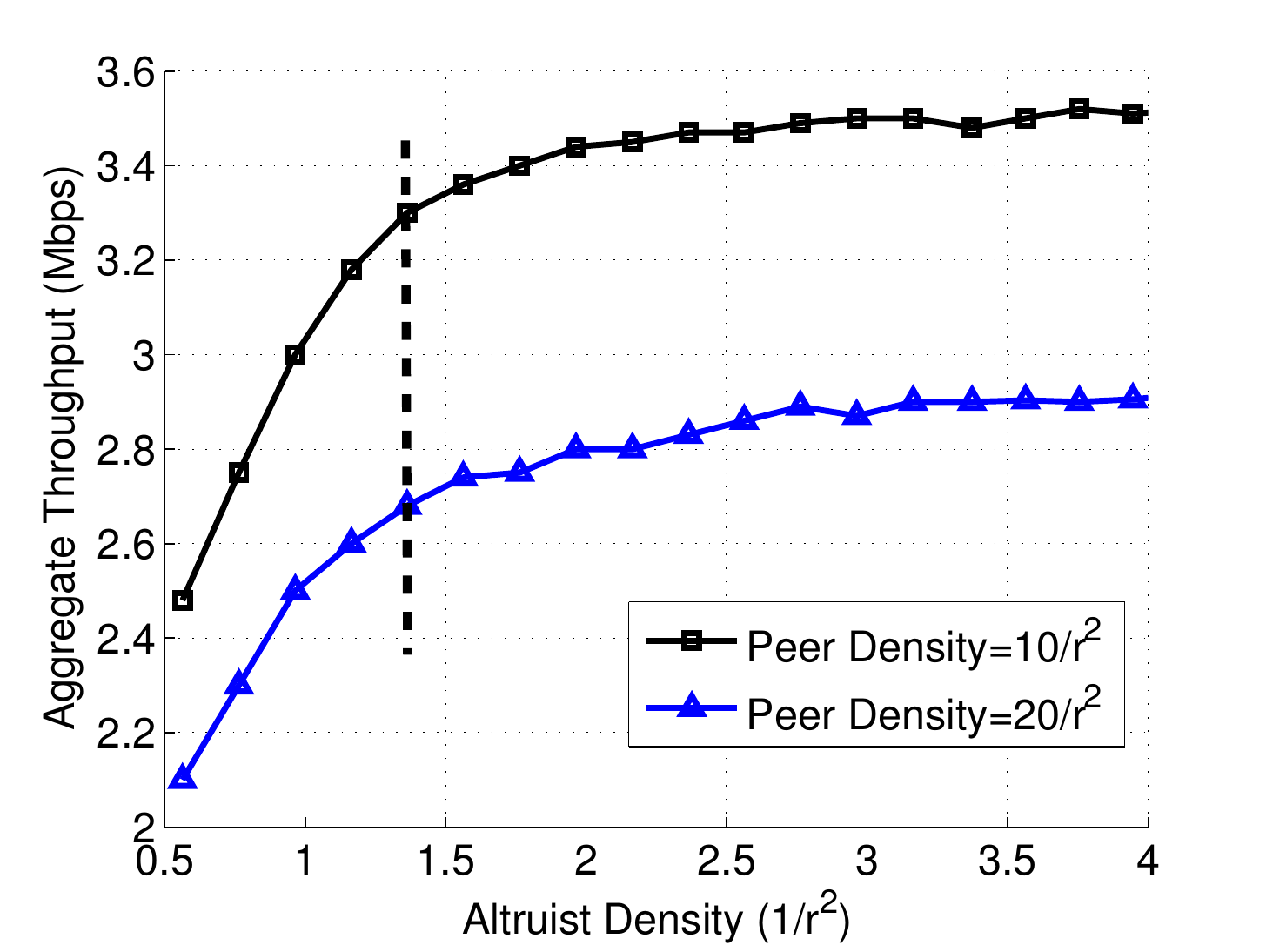}
    \label{fig:fman_thpt}}
    \subfloat[Aggregate power consumption.]{\includegraphics[trim=2mm 0 1.1cm 6mm, clip, width=0.48\linewidth]{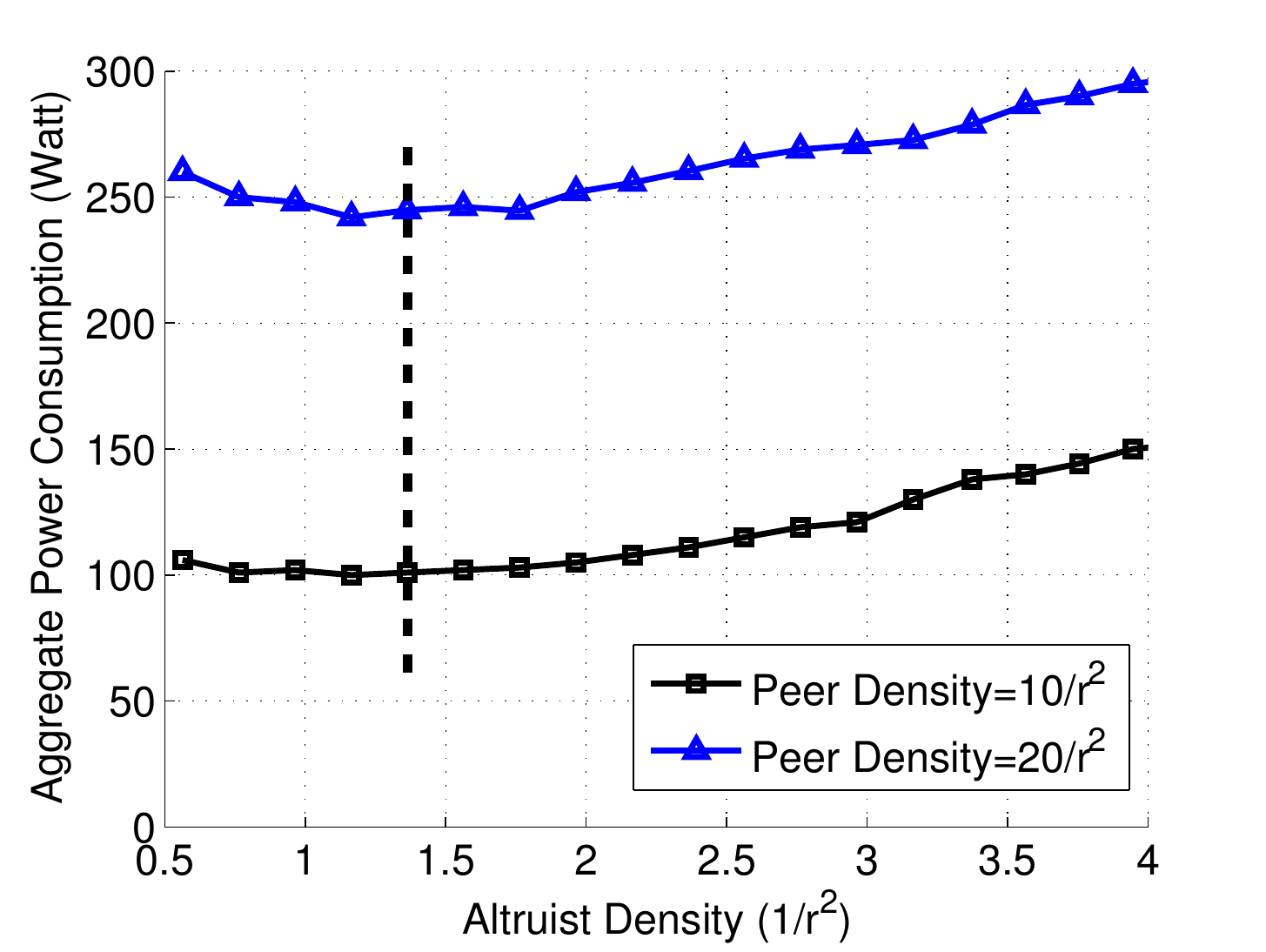}
    \label{fig:fman_ener}}
\else
    \subfloat[Aggregate throughput.]{\includegraphics[trim=4mm 0 1.1cm 6mm, clip, width=0.6\linewidth]{fman_thpt}
    \label{fig:fman_thpt}}\vfil
    \subfloat[Aggregate power consumption.]{\includegraphics[trim=2mm 0 1.1cm 6mm, clip, width=0.6\linewidth]{fman_ener}
    \label{fig:fman_ener}}
\fi
\caption{Finding the optimal altruist density for {\sl Altruistic}.}
\label{fig:fireperf}
\end{figure}

To investigate the performance of altruistic DISH around $\rho_{alt}$=1.31/$r^2$, we run {\sl Altruistic} in multi-hop networks by varying $\rho_{alt}$ from 0.56/$r^2$ to slightly more than 3.75/$r^2$, which corresponds to varying $p_{cov}$ from 50\% to more than 99\%. Data generation (at each peer) is Poisson with rate 25kbps. The results for peer density ($\rho_{peer}$) of $10/r^2$ and $\rho_{peer}=20/r^2$ are shown in \fref{fig:fireperf} (the results for $5/r^2$ and $30/r^2$ have the similar trend and are omitted). We see that, irrespective of the value of $\rho_{peer}$, the increasing throughput starts to level off at a knee point at $\rho_{alt}$ of 1.3--2$/r^2$ (\fref{fig:fman_thpt}), and the power consumption achieves the minimum also at $\rho_{alt}$ of 1.3--2$/r^2$ (\fref{fig:fman_ener}). This observation suggests a judicious choice of $\rho_{alt}$ in this range. Hence we adopt $\rho_{alt}=1.31/r^2$ as a near-optimum value which corresponds to $p_{cov}=80\%$. 

The results are explained as follows. Adding altruists converts UPs into CUPs and thereby reduces collisions and retransmissions. This helps increase throughput and save energy as well. On the other hand, as more and more altruists are added, more and more UPs become {\em redundantly} covered (by more than one altruists), meaning that the growth of $p_{cov}$ will slow down. This leads to the leveling off of throughput. In addition, since adding altruists contributes to a linear increase of energy consumption, the energy consumption starts to rise and deviate from the minimum.

\subsection{Arbitrary Deployment}\label{sec:deploy_arbi}

In arbitrary deployment, altruists can be carefully placed on a given topology formed by peers.

\begin{thm}\label{thm:nphard}
Consider a network with a given topology formed by peers on a finite plane. The problem of determining the minimum number and the locations of altruists to achieve full cooperation coverage, is NP-hard.
\end{thm}

\begin{IEEEproof}See Appendix.\end{IEEEproof}
\thispagestyle{appendixpage}

We remark on how to solve this problem in practice. In our proof, we have converted this problem into the classic set cover problem\cite{karp72} which has approximate solutions using a number of greedy algorithms (see book \cite{intro2alg2ed-greedy}). In our particular case (node deployment), these algorithms can be executed offline and hence do not introduce any runtime overhead. Regarding the performance of these algorithms, Alon et al. \cite{alon06} have recently established a lower bound to the approximation ratio, that such a greedy algorithm can achieve in polynomial time, to be $c \cdot \ln n$, where $c$ is a constant coefficient and $n$ is the number of elements to cover (i.e., UPs in our case).

A plausible thought is that we can carefully deploy altruists to cover the {\em entire region} and thereby achieve full cooperation coverage irrespective of the topology of peers. We can show that the minimum number of altruists to cover a rectangular area of $w\times h$ is $\lceil w/\sqrt{2}r \rceil \cdot \lceil h/\sqrt{2}r \rceil$. However, this argument is not true because covering an entire region is not equal to covering each UP (of two peers by a common altruist).

\section{Cost Efficiency}\label{sec:costeff}

We propose a metric called bit-meter-price (BMP) ratio to measure the cost efficiency of a protocol.

\subsection{Bit-Meter-Price Ratio}\label{sec:bmpdef}

BMP is a performance metric defined for a network:
\begin{align}\label{eq:bmp}
BMP \triangleq \frac{\overrightarrow{F}\cdot \overrightarrow{D}\cdot b_0}{ (N_p+N_a)\cdot \max(P_p^{max}, P_a^{max})},
\end{align}
where $\overrightarrow F$ is a vector of all the flows' throughput, $\overrightarrow{D}$ is a vector of all the flows' source-to-destination Euclidean distances, $N_p$ and $N_a$ are the total number of peers and altruists, respectively, $P_p^{max}$ and $P_a^{max}$ are the maximum power consumption rate among all the peers and the altruists, respectively, $b_0=e_0/c_0$, and $e_0$ and $c_0$ are the initial energy and the unit cost of a node (altruists and peers are the same devices), respectively.

BMP can be understood as
\[ \frac{Throughput (F) \cdot Distance (D) \cdot Lifetime (L)}{Price (C)},\]
where
\begin{align}
L &\triangleq \frac{e_0}{\max(P_p^{max}, P_a^{max})},\label{eq:lifetime-def}\\
C &\triangleq c_0 \cdot (N_p + N_a). \nonumber
\end{align}
In words, BMP is the total amount of successfully delivered data multiplied by end-to-end distance during the network's operational time and normalized by system resources. So the higher BMP, the better performance. The unit of BMP is bit$\cdot$m$/$\$.

In \eqref{eq:lifetime-def}, lifetime is defined as the time until any node (a peer or an altruist) runs out of energy. As peers are the nodes who actually perform the essential task of a network (transferring data), it also makes sense to define lifetime in terms of peers only, viz., $L = e_0/P_p^{max}$. It is easy to see this alternative definition is to the favor of {\sl Altruistic} because it only leads to a {\em higher} BMP for {\sl Altruistic}. But we still use definition \eqref{eq:lifetime-def} in our study.

The applicable traffic patterns of BMP include many-to-one (tree), many-to-multiple (mesh), and many-to-many (ad hoc). As such, its intended applications broadly cover data collection, Internet access, conferencing, p2p communication, file transfer, etc.  BMP can be applied to networks of various topologies and spanning any (regular or irregular) plane areas (which are accounted for by $\overrightarrow{D}$), networks with different node models and numbers of nodes (accounted for by $e_0$, $c_0$ and $N_p+N_a$),\footnote{It is also easy to extend \eqref{eq:bmp} to accommodate for a heterogeneous network which contains multiple node models, by using aggregative summation.} and networks irrespective of single- or multi-channel, single- or multi-hop. For networks without altruists, simply set $P_a^{max}=0$ and $N_a=0$. Ultimately, BMP may be generally used to evaluate cost efficiency for various protocols in various scenarios.

\begin{figure}[tb]
\centering
\ifdefined\thesis
    \subfloat[Multi-hop networks.]{
    \includegraphics[width=0.48\linewidth]{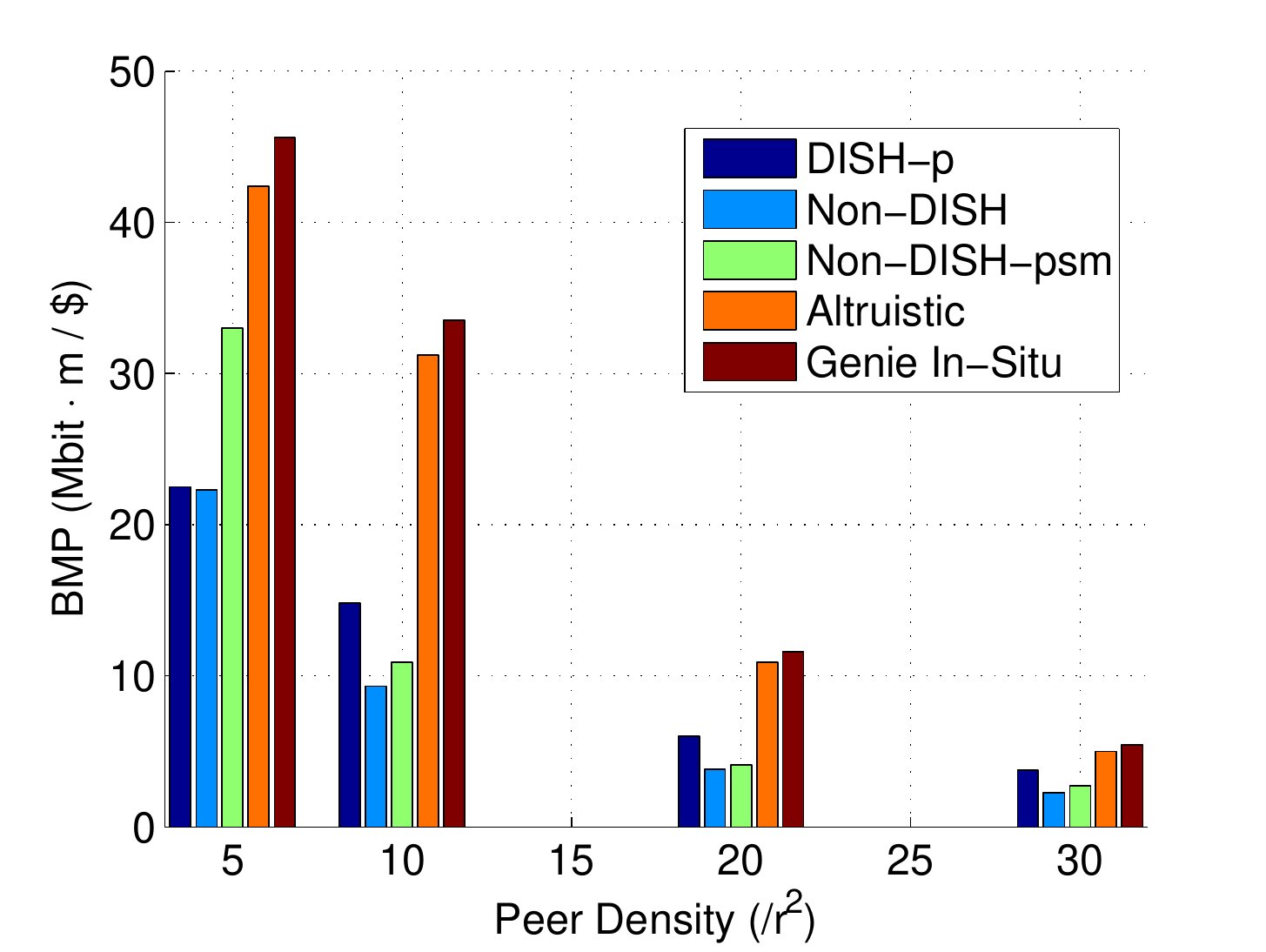}
    \label{fig:bmp_mt}}
    \subfloat[Single-hop networks.]{
    \includegraphics[width=0.48\linewidth]{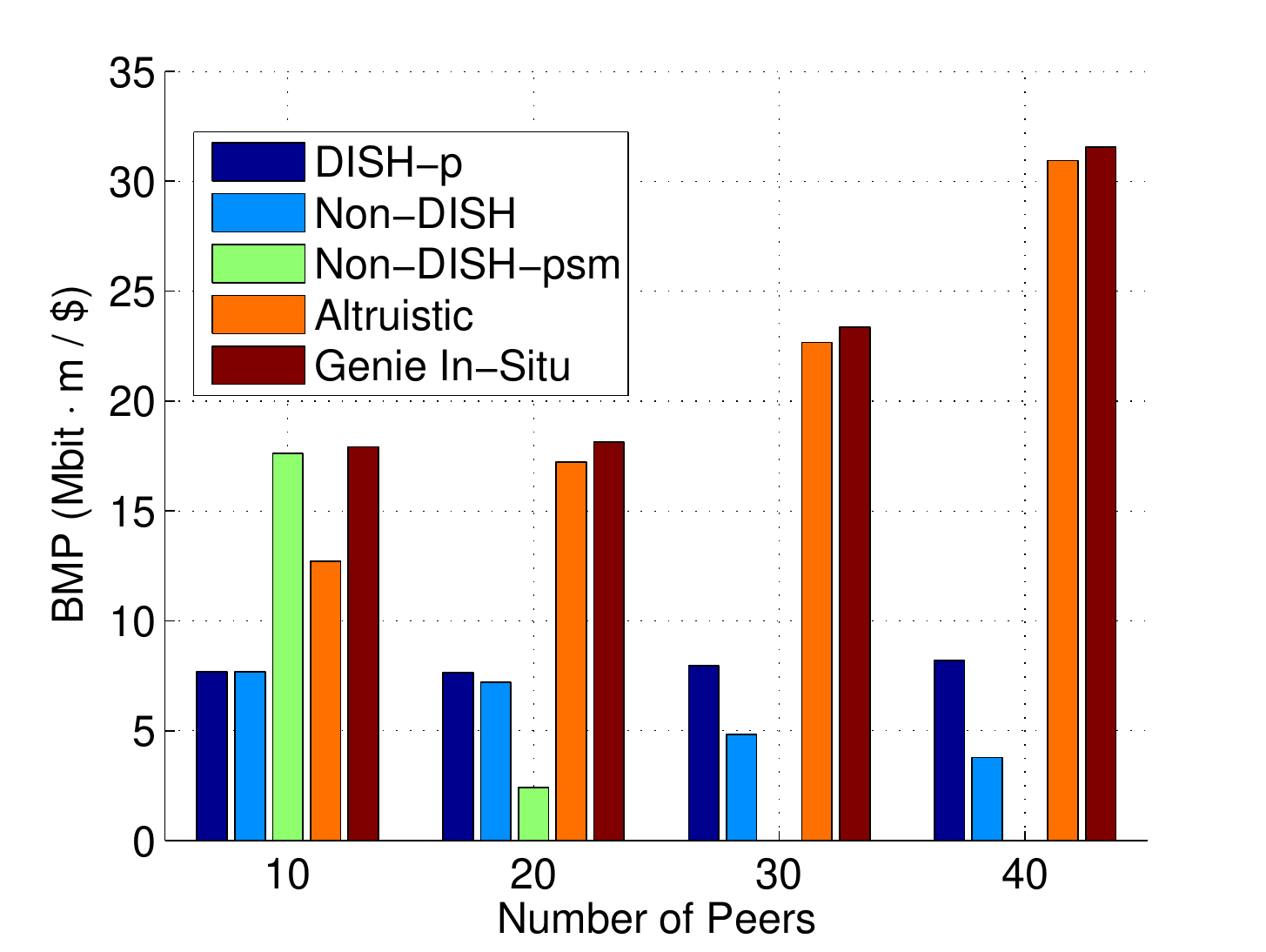}
    \label{fig:bmp_sg}}
\else
    \subfloat[Multi-hop networks.]{
    \includegraphics[width=0.65\linewidth]{bmp_mt}
    \label{fig:bmp_mt}}\vfil
    \subfloat[Single-hop networks.]{
    \includegraphics[width=0.65\linewidth]{bmp_sg}
    \label{fig:bmp_sg}}
\fi
\caption{Cost efficiency evaluated via BMP. The higher, the better.}
\label{fig:BMP}
\end{figure}

\subsection{BMP Evaluation}\label{sec:bmpeval}

We conduct simulation and compute BMP for the five protocols. Since all the protocols use the same devices, the value of $b_0$ does not affect comparison and we set $b_0=1$J/\$. For {\sl Altruistic}, we deploy altruists with density $\rho_{alt}=1.31/r^2$ in multi-hop networks according to \sref{sec:deploy_rand}, and deploy a single altruist in single-hop networks which achieves full cooperation coverage. Each source node generates data at 25kbps in multi-hop networks and 160kbps in single-hop networks.

The results are shown in \fref{fig:BMP}. We see that, apart from {\sl Genie In-Situ}, {\sl Altruistic} is the clear winner among all the protocols: its BMP is more than twice the BMP of the other protocols in most cases. Compared to {\sl Genie In-Situ}, the BMP of {\sl Altruistic} is only slightly lower. In fact, for a {\em real} in-situ energy conscious DISH protocol (without a genie), the complexity and overhead for rotating the responsibility of cooperation, as discussed in \sref{sec:qualitative}, would negate this marginal advantage of {\sl Genie In-Situ} over {\sl Altruistic}.

Here we also provide an intuitive understanding of how {\sl Altruistic} performs well. We inspect each component of BMP for {\sl Altruistic} and {\sl DISH-p} at the peer density of 10/$r^2$ in \fref{fig:bmp_mt}, as an example.
\begin{itemize}
\item $Throughput \cdot Distance$: measured to be 3826 Mbit$\cdot$m/s for {\sl DISH-p} and 3822 Mbit$\cdot$m/s for {\sl Altruistic}. These two values are almost equal, which indicates that, since $\overrightarrow{D}$ is statistically the same for the two protocols, a cooperation coverage of 80\% ($\rho_{alt}=1.31/r^2$) suffices to achieve a cooperation gain (in terms of throughput) similar to that achieved by the opportunistic cooperation in {\sl DISH-p}.\footnote{A theoretical analysis of the probability of obtaining cooperation in DISH-p can be found in \cite{tie08mobihoc,tie10tmc-analysis}.}
\item $Lifetime$: lifetime of {\sl Altruistic} ($e_0/0.718Watt$) is 2.385 times that of {\sl DISH-p} ($e_0/0.301Watt$)---peers can sleep due to the existence of altruists.
\item $Price$: {\sl Altruistic} uses 407 nodes which is 13\% more than what {\sl DISH-p} uses (360 nodes).
\end{itemize}
Eventually, BMP of {\sl Altruistic} is 31.2 Mbit$\cdot$m/\$ and BMP of {\sl DISH-p} is 14.8 Mbit$\cdot$m/\$, which translates to a significant ratio of 2.11.

Finally, we explain the different trends that the five protocols exhibit in different scenarios. In \fref{fig:bmp_mt} (multi-hop networks), the BMP declines as the number of peers increases. This is because of the drop of throughput (exponentially) as established by \cite{gupta00}, the drop of lifetime due to more energy consumed in packet transmission and channel contention, and the increase of cost. In single-hop networks (\fref{fig:bmp_sg}), (1) the BMP of {\sl Non-DISH} gradually declines due to the lack of information sharing, (2) the BMP of {\sl Non-DISH-psm} drops remarkably due to the lack of both information sharing and gathering, (3) the BMP of {\sl DISH-p} largely maintains, and (4) the BMP of {\sl Altruistic} and that of {\sl Genie In-Situ} both rise as by virtue of energy conservation of the strategies as well as the throughput benefit of DISH.

In summary, the evaluation of cost efficiency demonstrates that the additional cost of altruists pays off; the performance gain from altruistic DISH more than offsets the marginal cost increase.

\section{Throughput-Energy Trade-off}\label{sec:perf}

This section zooms in to specifically inspect the throughput and energy performance.

\subsection{Multi-Hop Networks}\label{sec:mthop}

\begin{figure}[tb]
\centering
\ifdefined\thesis
    \subfloat[Throughput.]{\includegraphics[width=0.45\linewidth]{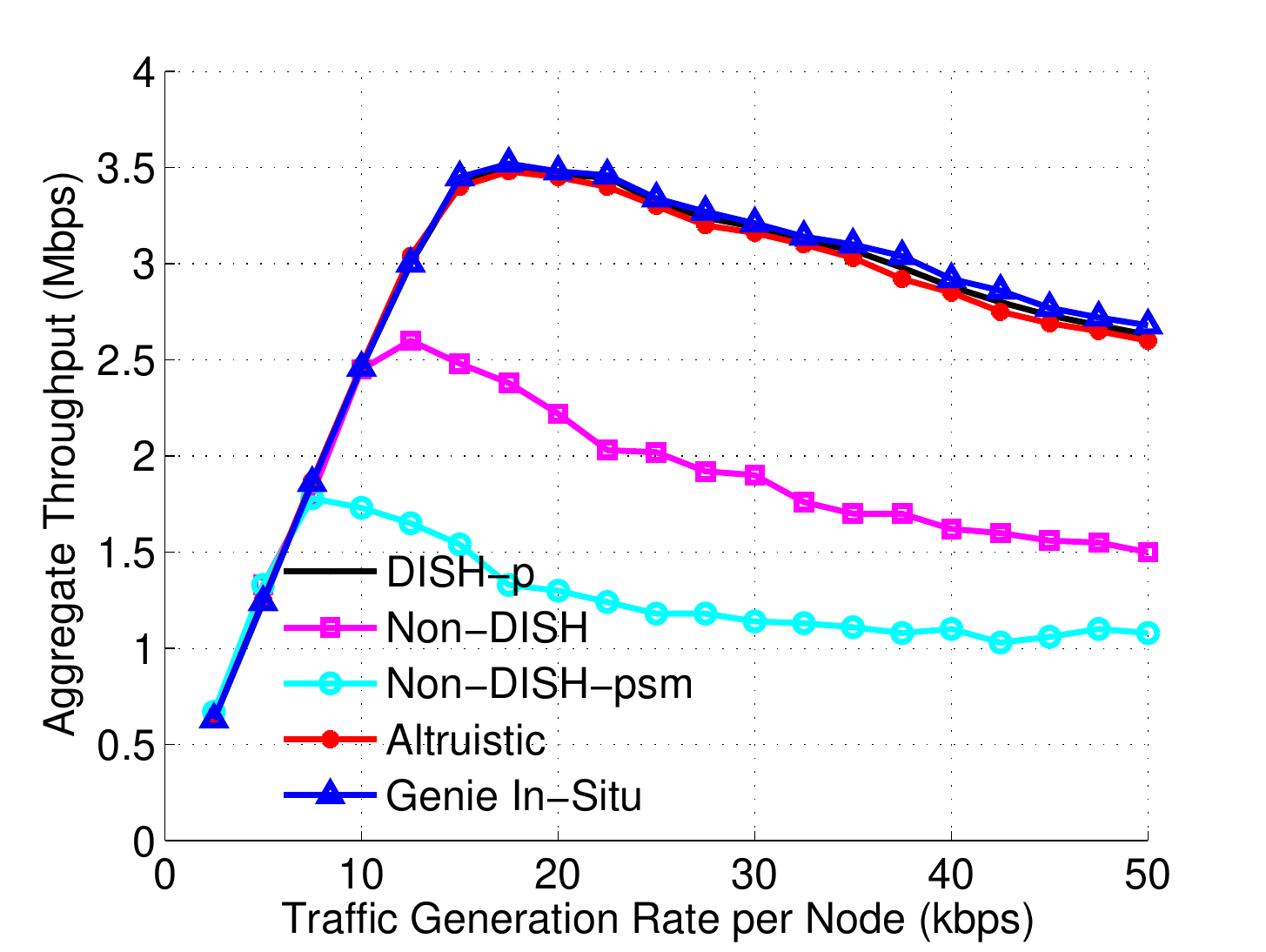}
    \label{fig:mt_thpt}   }
    \subfloat[Power consumption.]{\includegraphics[width=0.45\linewidth]{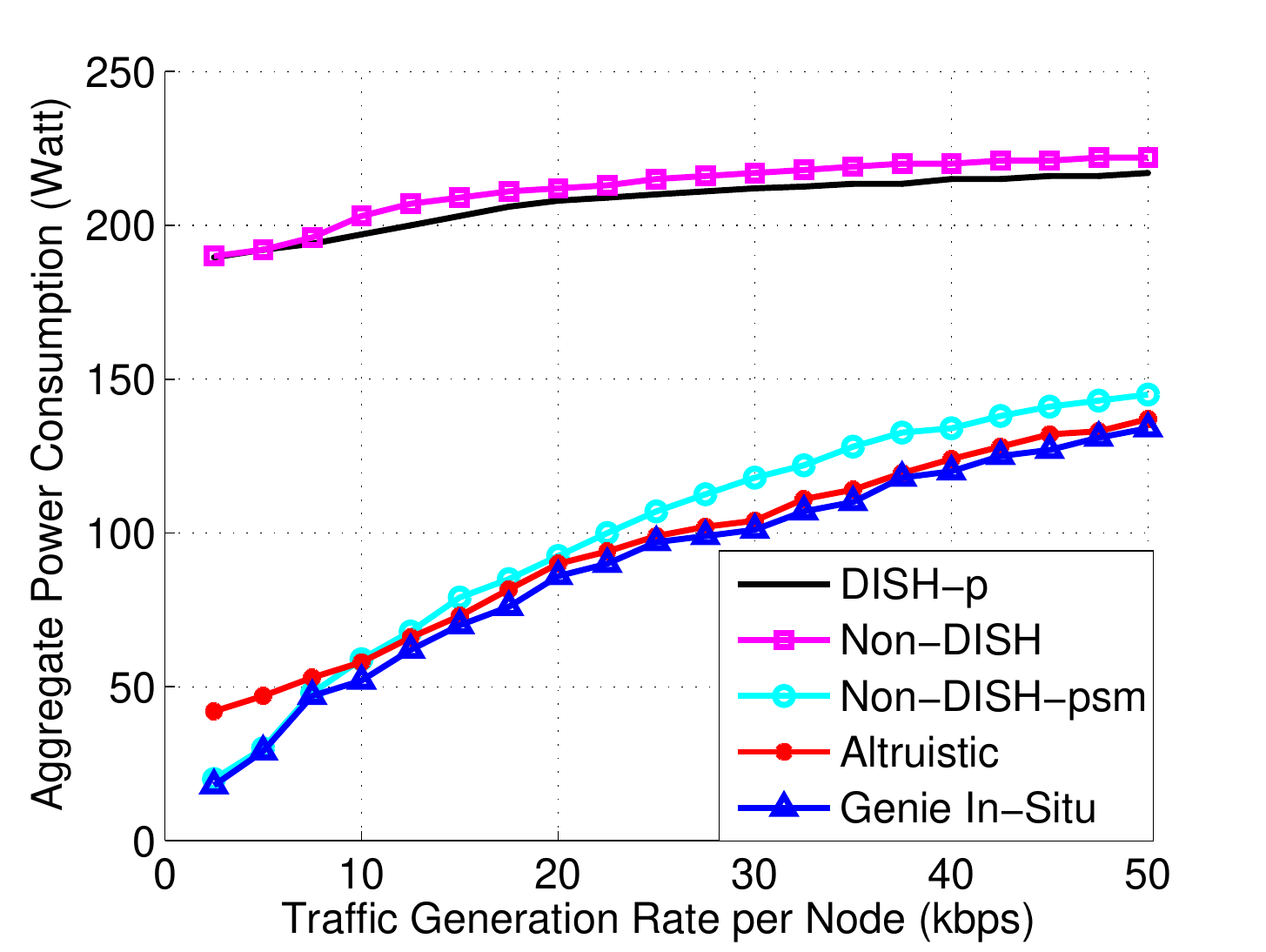}
    \label{fig:mt_ener}   }
\else
    \subfloat[Throughput.]{\includegraphics[width=0.65\linewidth]{mt_thpt_80_10-all}
    \label{fig:mt_thpt}   }\vfil
    \subfloat[Power consumption.]{\includegraphics[width=0.65\linewidth]{mt_ener_80_10-all}
    \label{fig:mt_ener}}
\fi
\caption{Throughput-energy tradeoff in multi-hop networks.}
\label{fig:mthop}
\end{figure}

The simulation setup remains the same, and the results are shown in \fref{fig:mthop} for $\rho_{alt}=1.31/r^2$ and $\rho_{peer}=10/r^2$ (the results for $\rho_{peer}=20/r^2$ are similar and omitted). \fref{fig:mt_thpt} (throughput) clearly indicates three levels as low, medium, and high, corresponding to {\sl Non-DISH-psm}, {\sl Non-DISH}, and the three DISH protocols ({\sl DISH-p}, {\sl Genie In-Situ} and {\sl Altruistic}), respectively. For example at the traffic generation rate of 25kbps, {\sl Non-DISH} achieves 64\% higher throughput than {\sl Non-DISH-psm}, and the three DISH protocols achieve 65\% higher than {\sl Non-DISH}. This is readily explained by the use of information gathering and/or sharing. The main message to take away from this set of results, however, is that both of the two energy-efficient strategies can {\em preserve} the throughput benefit of DISH.

For power consumption as shown in \fref{fig:mt_ener}, we see that both {\sl Altruistic} and {\sl Genie In-Situ} save a remarkable amount (40--80\%) of energy consumed by {\sl DISH-p} or {\sl Non-DISH}. Noteworthily, {\sl Altruistic} even outperforms {\sl Non-DISH-psm} (though slightly) under higher traffic load, which is somehow counter-intuitive because {\sl Non-DISH-psm} seems to be the most energy-frugal protocol where all nodes sleep whenever possible, and {\sl Altruistic} has additional nodes who are always awake. In fact, the amount of energy saved by the altruists (through avoiding collisions and retransmissions caused by MCC problems) becomes more significant under higher traffic, where MCC problems are created more often, and outweighs the energy consumed by these few altruists.

\subsection{Single-Hop Networks}\label{sec:sghop}
\begin{figure*}[tb]
\centering
\subfloat[Throughput (saturated traffic).] {\includegraphics[width=0.33\textwidth]{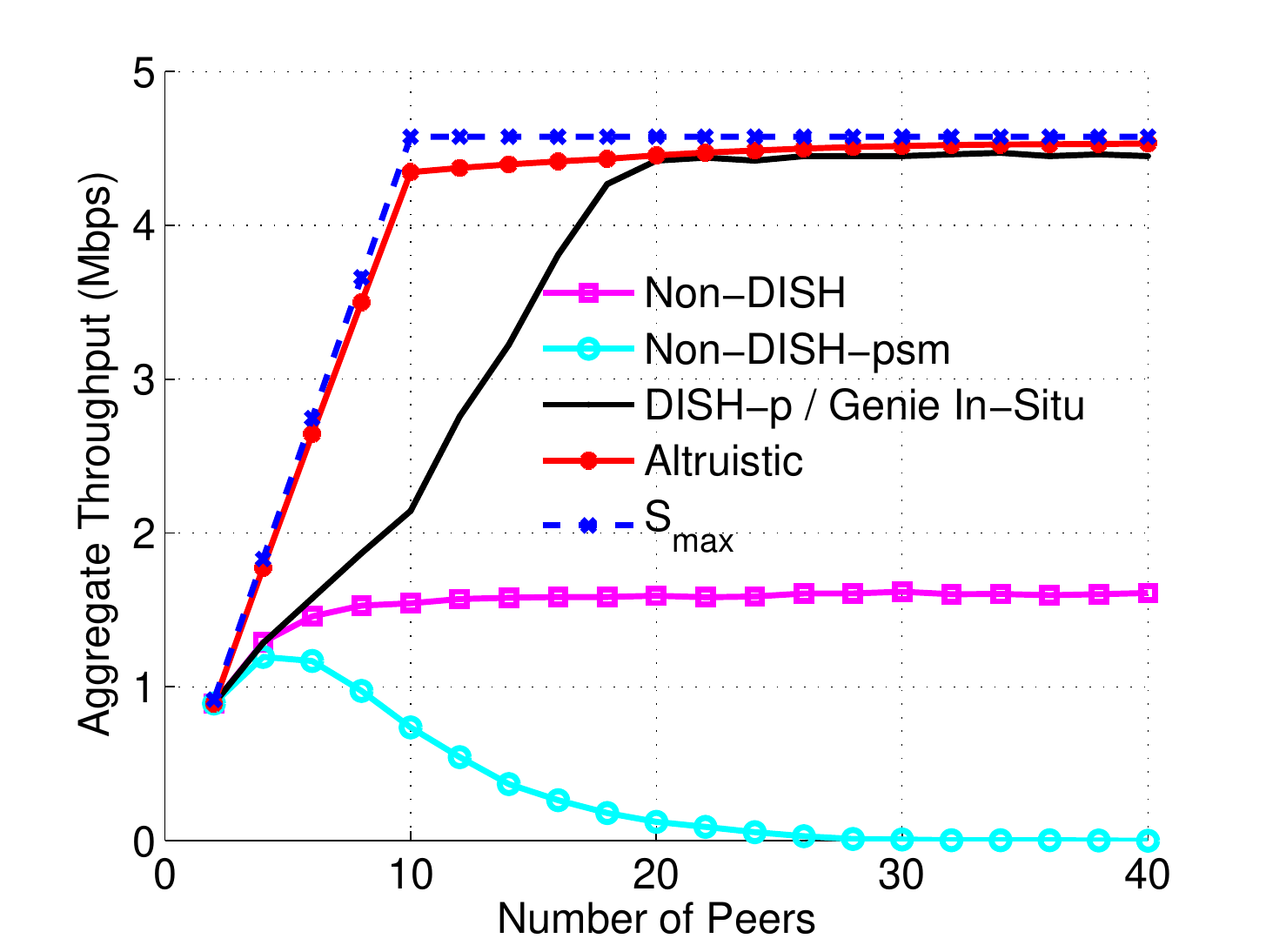}
\label{fig:sg_thpt_highload}}
\subfloat[Power consumption (saturated traffic).]
{\includegraphics[width=0.33\textwidth]{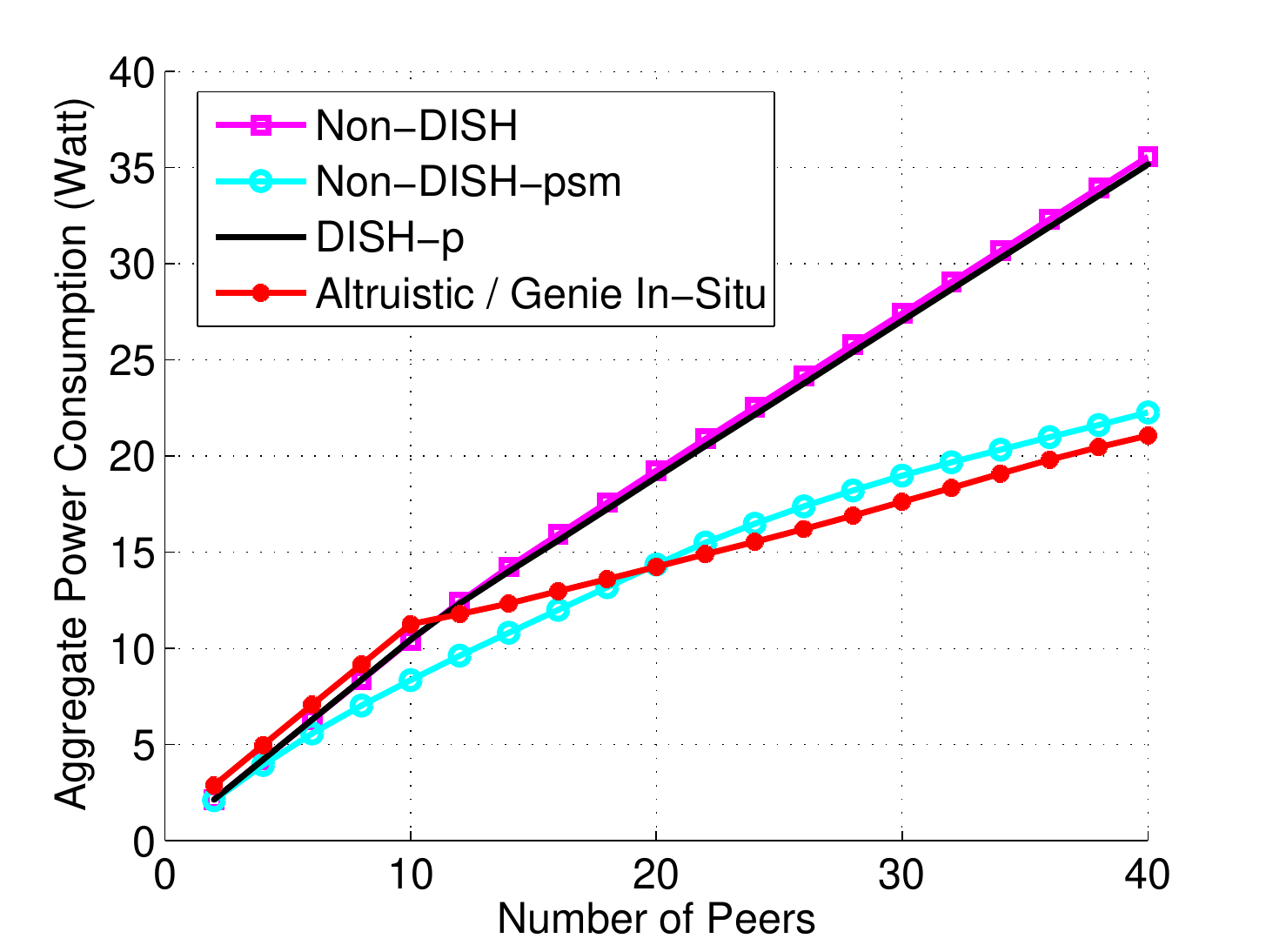}\label{fig:sg_ener_highload}}
\subfloat[Power consumption (light traffic).]
{\includegraphics[width=0.33\textwidth]{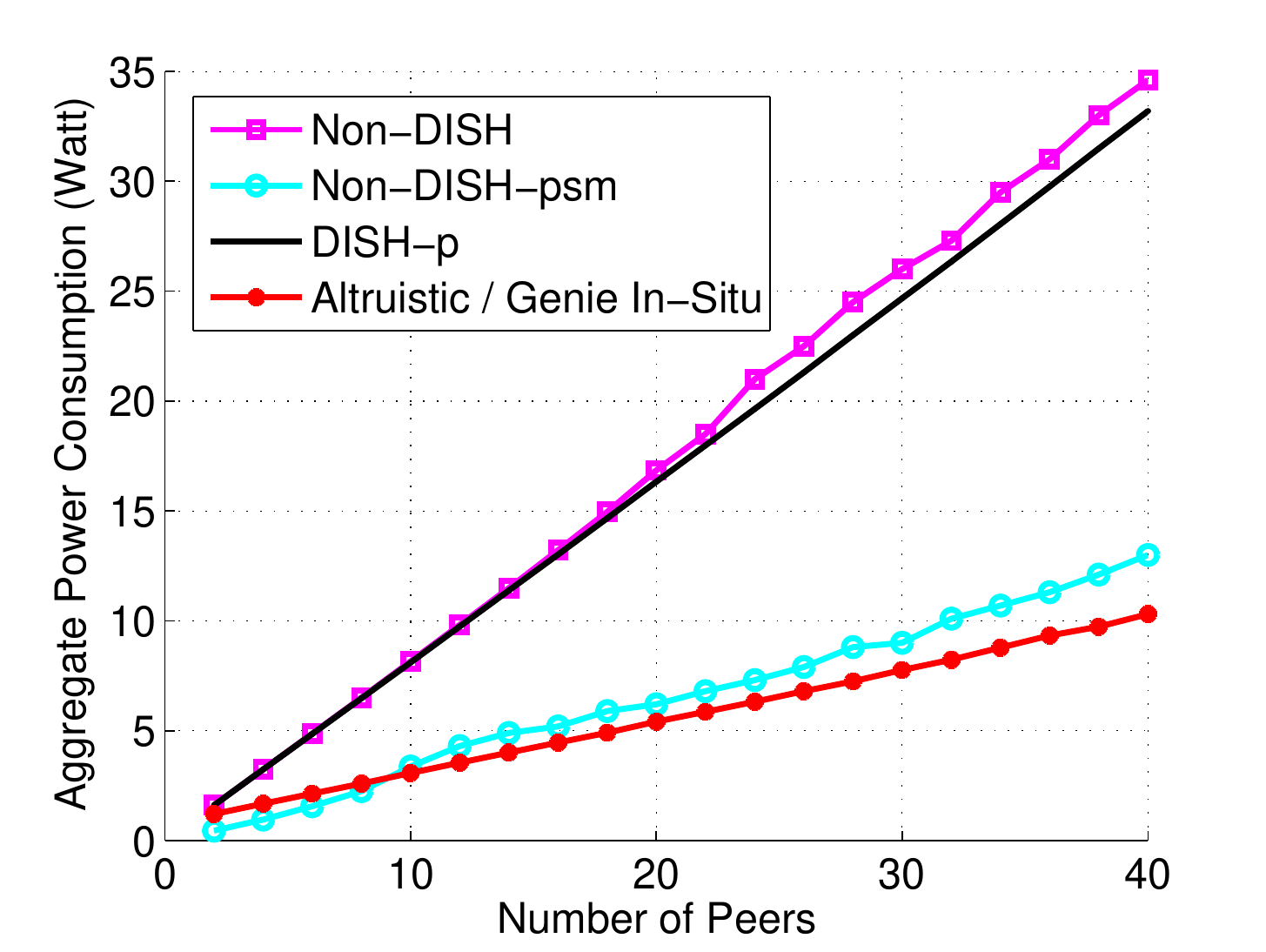}\label{fig:sg_ener_lowload}}
\caption{Throughput-energy tradeoff in single-hop networks. Some curves overlap almost completely and so are plotted as one curve for clearer visualization, as can be seen in the legend.}
\label{fig:sghop}
\end{figure*}

{\sl Altruistic} uses one altruist in single-hop networks. The simulation was conducted under high traffic load (source nodes are always backlogged) and low traffic load (traffic generation rate is 160kbps) respectively, and the results are summarized in \fref{fig:sghop}. For throughput shown in \fref{fig:sg_thpt_highload}, other than observing the similar gaps to \fref{fig:mt_thpt}, we notice that {\sl Altruistic} outperforms {\sl DISH-p} and even {\sl Genie In-Situ} when the number of peers is less than 20. This is because, when peers are few and traffic load is high, peers will stay on data channels most of the time and lead to {\sl DISH-p} and {\sl Genie In-Situ} lacking of cooperative nodes (who must be on the control channel). However, {\sl Altruistic} has a {\em dedicated} cooperative node and does not face this problem at all. Another observation is that {\sl Altruistic} closely approaches $S_{max}$, a theoretical throughput upper bound:
\begin{align}\label{eq:upbound}
S_{max}=\frac{\min(m,n_f)\cdot T_{payload}\cdot W}{T_{cca}^{min}+T_{ctrl}+T_{data}+T_{sw}},
\end{align}
where $m$ is the number of data channels, $n_f$ is the number of flows, $W$ is the data channel bandwidth, $T_{payload}$ is the transmission time of data payload, $T_{cca}^{min}$ is the minimum CCA duration, $T_{ctrl}$ and $T_{data}$ are the duration of a successful control/data channel handshake, and $T_{sw}$ is channel switching delay. The derivation of $S_{max}$ is given in \cite{tie09tmc}. Moreover, when there are more than 20 nodes, the throughput of {\sl Non-DISH-psm} is very low, because each data channel has more than 4 fully-loaded competing nodes on average (recall that there are 5 data channels) and hence is almost always busy. As nodes in {\sl Non-DISH-psm} do not gather information and always choose a channel from all channels, collision will happen for almost every channel use.

\fref{fig:sg_ener_highload} and \fref{fig:sg_ener_lowload} present the energy performance. Under both high and low traffic loads, {\sl Altruistic} conserves energy substantially. For example in the low-load scenario at 40 peers, it consumes only 30\% power of {\sl DISH-p}. In addition, {\sl Altruistic} again slightly outperforms {\sl Non-DISH-psm}, which has been explained in \sref{sec:mthop}.

In summary, the simulations demonstrate that altruistic DISH conserves a significant amount of energy and well maintains the throughput benefit of DISH.

\ifdefined\thesis
\else
\section{Hardware Implementation}\label{sec:testbed-energy}

We have also implemented four protocols on COTS hardware (all the five except {\sl Genie In-Situ} which requires a non-implementable genie). To the best of our knowledge, these are the first full implementation of asynchronous multi-channel MAC protocols for ad hoc networks (see review in \sref{sec:impl-relwork}).

\subsection{Implementation}\label{sec:impl}
\ifdefined\thesis
\subsection{Platform Selection}
\else
\subsubsection{Platform Selection}
\fi

We chose a micro-controller (MCU) based platform with an ASIC radio, instead of (i) an FPGA-based platform which was more expensive and required hardware description language (HDL) in programming, or (ii) a software radio whose MAC source code was not fully available. Among the ASIC radios, we chose 802.15.4 radios instead of 802.11 radios because 802.11-radio based devices (such as laptops and PDAs) have higher cost and bigger size than 802.15.4 devices, and 802.11-based development kits (such as HostAP \cite{hostap} and MadWifi \cite{madwifi} as used by \cite{sigcomm09dual}) have more limited MAC layer control than 802.15.4-based software (such as TinyOS~\cite{tinyos}). 

Eventually, we chose TelosB Mote~\cite{telos}, which is a MCU platform with an ASIC radio (CC2420~\cite{chipcon}) as our hardware platform and TinyOS 2.0 as our software platform.
TinyOS has almost full control over the MAC layer, and its component-based architecture and C-like programming language enable rapid development.
Note that such a platform choice should not be used to establish benchmarks for WiFi cards, though it suffices for a comparative study like ours.

\ifdefined\thesis
\subsection{Overcoming Limitations}\label{sec:fix}
\else
\subsubsection{Overcoming Limitations}\label{sec:fix}
\fi

There are two major limitations of the hardware we choose. First, the CC2420 radio supports packet size of up to only 127 bytes. We overcome this by substituting each data packet with a sequence of data fragments and treating the inter-fragment intervals as {\em payload}. In other words, let $n_{frag}$ be the number of fragments, $l_{frag}$ be the length of each fragment, and $\tau$ be the interval, then the transmission time of a data packet is
\begin{equation}\label{eq:datafrag}
t_{data} = n_{frag} (\frac{l_{frag}}{w} + \tau + t_d) - \tau
\end{equation}
where $w$=250kbps is the channel data rate, $t_d$ (100--200$\mu$s) is the latency that each fragment takes to be sent into the air after being assembled in memory.  The second limitation is that the timing accuracy of TelosB is not reliable at the microsecond level while reliable at the millisecond level. Thus we proportionally scale all protocol intervals up to milliseconds, e.g., SIFS is scaled to 2ms. This way, a control channel handshake lasts for $t_{ctrl}\approx9$ms. Now getting back to \eqref{eq:datafrag}, in order to keep the ratio $t_{data}:t_{ctrl}$ close to our simulation, we chose $n_{frag}=20$, $l_{frag}=30$bytes (including preamble), and $\tau$=8ms, and consequently $t_{data}\approx175$ms.

\ifdefined\thesis
\subsection{Virtual Collision Detection}
\else
\subsubsection{Virtual Collision Detection}
\fi

Interestingly, how we overcame the limitations described above enabled us to devise a simple yet accurate technique for {\em packet collision detection}. Collision detection is useful in many network algorithms (such as collision avoidance, flooding, channel selection, and data aggregation) \cite{emnets05capture}, but is non-trivial because a usual PHY layer cannot distinguish packet collision from {\em noise corruption}. Prior techniques generally use link quality indicator (LQI) and/or received signal strength indicator (RSSI). However, they are empirical and lack in accuracy, and according to \cite{zhao03sensys,mass04link,globe03link}, it is still controversial whether RSSI or LQI is a better indicator for link quality.

Our technique is {\em virtual collision detection} which achieves the goal using {\em interleaved} fragment sequences. The idea is based on the fact that each data packet is transmitted as a sequence of fragments and the fragment interval (8 ms) is much larger than the fragment transmission time ($<$1 ms). Therefore, a good indicator of data collision is an interleaved sequence of fragments which contains fragments sent by more than one senders (recall that intervals are counted as actual payload). \fref{fig:interleave} illustrates this.

\begin{figure*}[tb]
\centering
    \includegraphics[width=\linewidth]{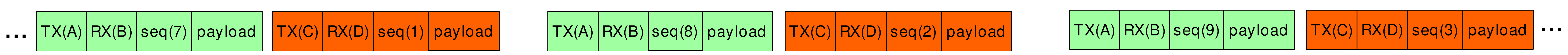}
\caption{Virtual collision detection. There are two interleaved fragment sequences, where \textsl{TX-RX}'s are {\em alternate} and \textsl{seq}'s are {\em inconsecutive}.}
\label{fig:interleave}
\end{figure*}

\subsection{Experiments}\label{sec:experiment}

For visualization purposes, we use the three LEDs on each TelosB mote to indicate specific events of interest (a maximum number of $2^3=8$ events can be represented). For example, a blue LED indicates an ongoing control channel handshake, a green LED indicates an ongoing data channel handshake, and a red LED indicates transmitting a cooperative message. Other events are indicated by LED combinations. \fref{fig:mote11} gives a snapshot in a trial indoor experiment.

\begin{figure}[tb]
\centering
\includegraphics[width=0.95\linewidth]{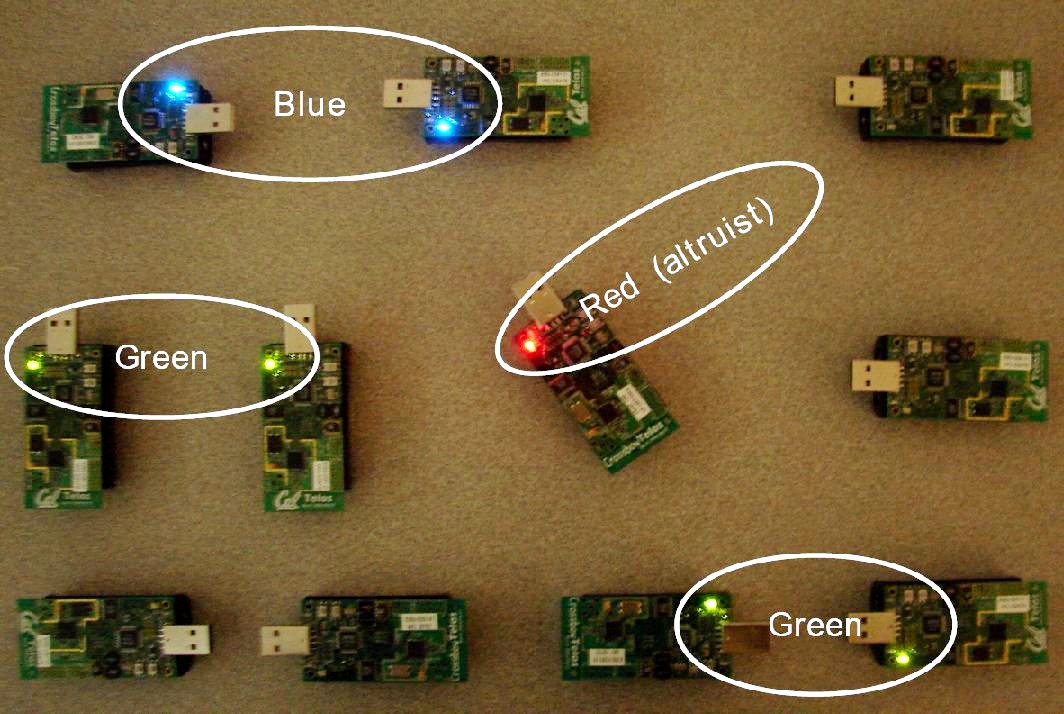}
\caption{A snapshot of a trial indoor experiment on {\sl Altruistic} with 11 nodes. The four ``green nodes'' are two sender-receiver pairs communicating on the two different data channels. A pair of ``blue nodes'' are performing a control channel handshake,
which creates a channel conflict problem because the only two data channels are already occupied. At this moment, the altruist (``red node'') identifies this
and sends a cooperative message (INV), which informs the blue nodes to back off and thereby avoid colliding with the two ongoing data transmissions.}
\label{fig:mote11}
\end{figure}

In our experiments, nodes are randomly placed in a 10m$\times$10m roof area, and the transmission power is set at 0dBm which is the maximum on CC2420.\footnote{With this setting, all nodes are within the radio range of each other, which was also used by \cite{so06sync,mcmac07wcnc,realman06hybrid}. To do multi-hop experiments, a large number of nodes are needed to demonstrate the impact of a small $\rho_{alt}$ on a large $\rho_{peer}$ as shown in \sref{sec:deploy_rand} and \sref{sec:mthop}.}
Nodes are configured as disjoint flows and source nodes are always backlogged. There are three channels (one control channel and two data channels) and each is with data rate 250kbps. To compute power consumption, we trace the TX/RX/IDLE states on each node to accumulate its sojourn time for each state, and at the end of each experiment, do a weighted sum using the same power consumption rates as in the simulation setup. For protocols using power saving mode, IDLE is treated as SLEEP where peers do not overhear. Alternatively, one can put motes actually into sleep by, e.g., developing a multi-channel version of B-MAC\cite{bmac04sensys} or X-MAC\cite{xmac06sensys} so as to measure the actual battery drainage. However, measuring the energy consumption of sensor nodes accurately is not only difficult~\cite{ymac08ipsn}, but also not necessary because (1) sensor nodes have a different energy model from WiFi nodes, (2) it requires using a real sleep-wake scheduling algorithm which will lose generality as explained in \sref{sec:proto-energy}, and (3) this is a comparative study and the goal is not to establish absolute-value benchmarks for TelosB. Our approach of computing energy consumption was also used by \cite{ymac08ipsn}.

In collecting statistics for the four protocols, every data point is by averaging over 8 experiments and each experiment runs for 600 real seconds.

\begin{figure}[tb]
\centering
\ifdefined\thesis
\includegraphics[width=0.55\linewidth]{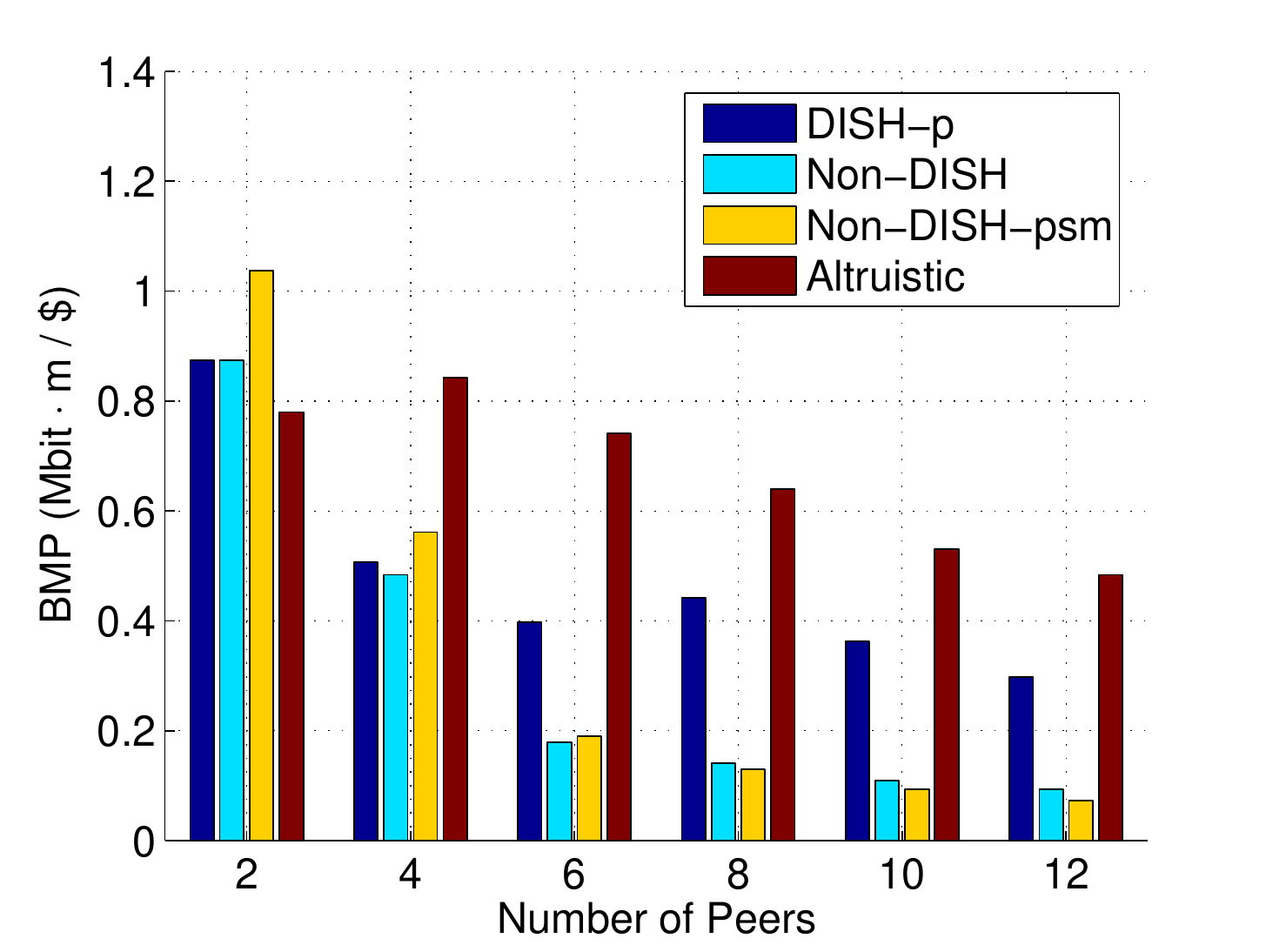}
\else
\includegraphics[width=0.7\linewidth]{exp_bmp}
\fi
\caption{Experimental results of cost efficiency.}
\label{fig:exp_bmp}
\end{figure}

\fref{fig:exp_bmp} summarizes the experimental results of cost efficiency, which confirms {\sl Altruistic} to be the clear winner among the protocols. The only exception appears when there are only two peers where {\sl Non-DISH-psm} performs the best. The reason is simply that DISH does not help in this contention-less case where there is only one sender-receiver pair, and that adding an altruist only increases cost and energy consumption.

\begin{figure}[tb]
\centering
\ifdefined\thesis
    \subfloat[Throughput.]{\includegraphics[width=0.45\linewidth]{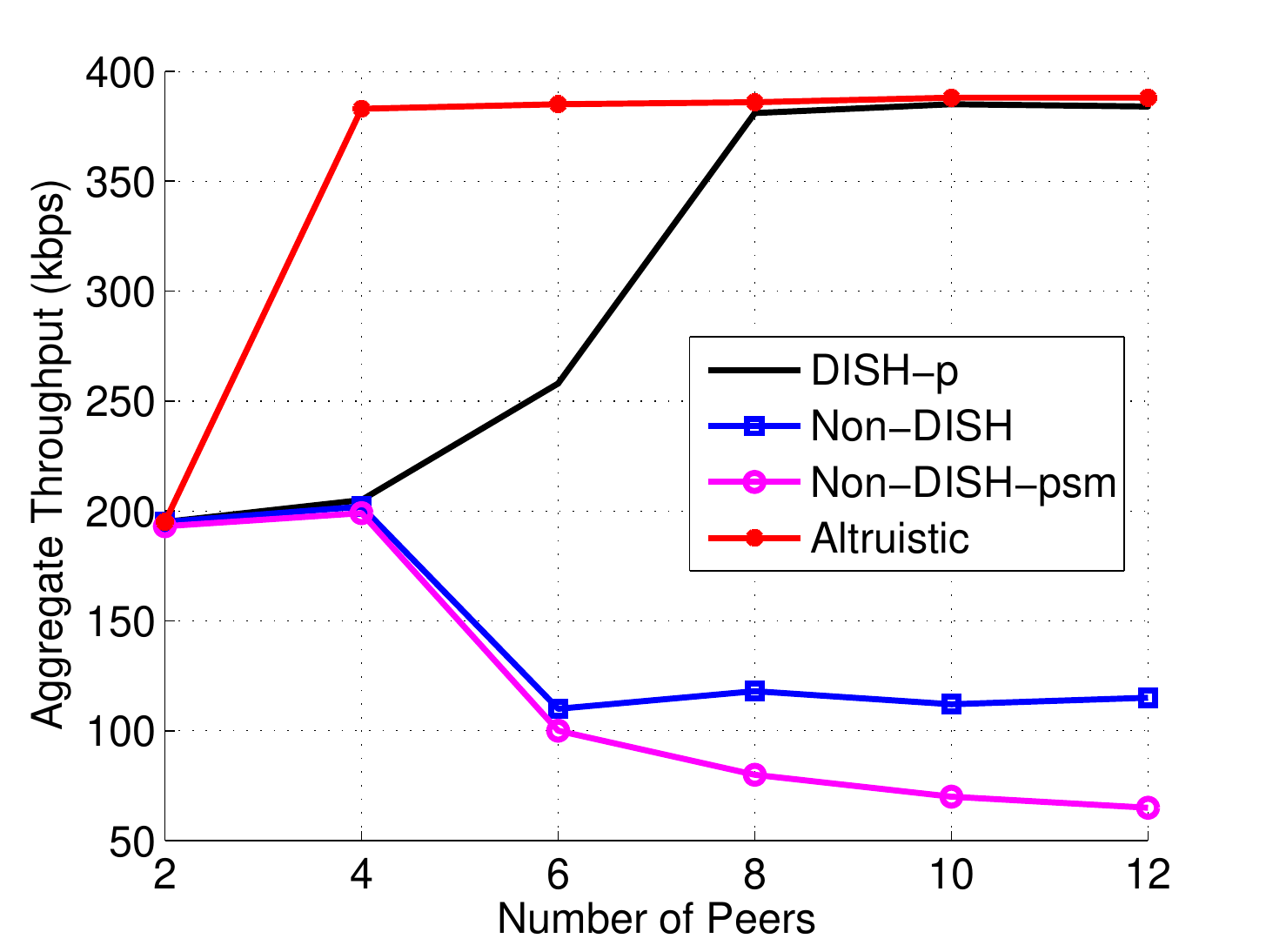}
    \label{fig:exp_thpt}}
    \subfloat[Power consumption.]{\includegraphics[width=0.45\linewidth]{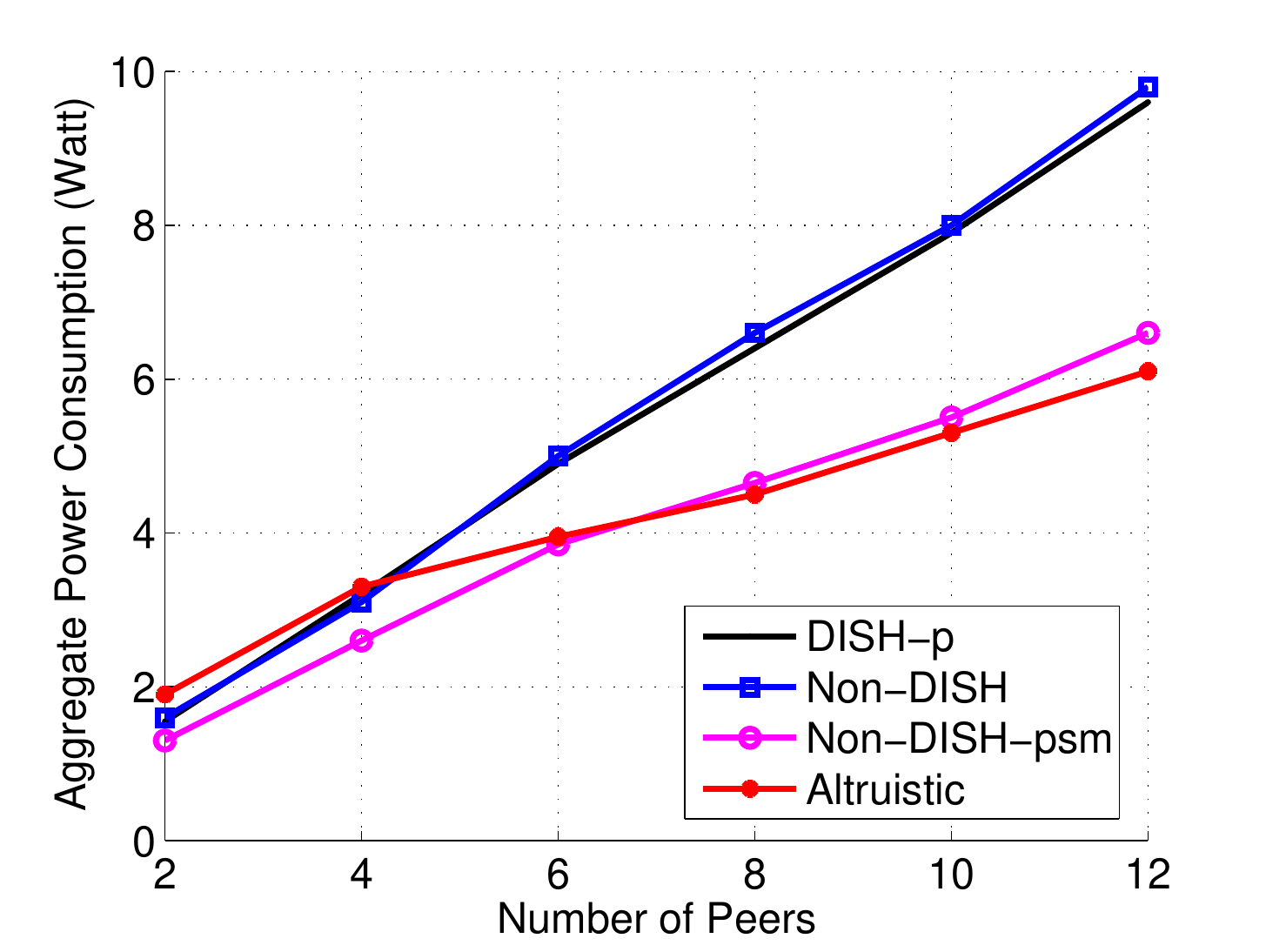}
    \label{fig:exp_ener}}
\else
    \subfloat[Throughput.]{\includegraphics[width=0.65\linewidth]{exp_thpt}
    \label{fig:exp_thpt}}\vfil
    \subfloat[Power consumption.]{\includegraphics[width=0.65\linewidth]{exp_ener}
    \label{fig:exp_ener}}
\fi
\caption{Experimental results of throughput-energy tradeoff.}
\label{fig:exp-energy}
\end{figure}

\fref{fig:exp-energy} gives the experimental results for throughput-energy tradeoff. We specifically used two data channels in order to see different trends from simulation rather than merely produce a scaled version of simulation. For throughput shown in \fref{fig:exp_thpt}, comparing it with \fref{fig:sg_thpt_highload} (simulation), we see that {\sl Non-DISH} and {\sl Non-DISH-PSM} in \fref{fig:exp_thpt} both have a sharp drop (by about 50\%) when the number of peers is 6, while in \fref{fig:sg_thpt_highload}, the throughput of {\sl Non-DISH} keeps increasing until finally saturates and the throughput of {\sl Non-DISH-PSM} {\em gradually} decreases. This difference from simulation arises from the peculiarity when there are only two data channels and 3 fully-loaded node pairs, as illustrated in \fref{fig:2channel}. 

The key message conveyed by \fref{fig:exp_thpt} is that {\sl Altruistic} still performs the best and, particularly, better than {\sl DISH-p} when the number of nodes is small, due to the guaranteed provision of cooperation.

\begin{figure}[tb]
\centering\includegraphics[trim=1.3cm 8cm 1.2cm 8cm, clip, width=\linewidth]{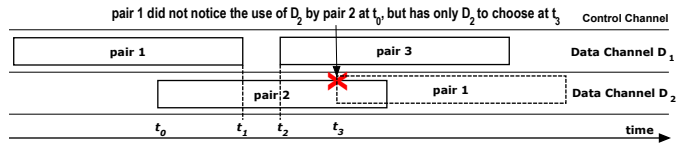}
\caption{The peculiarity when there are only two data channels. After pair 1 switching back to control channel at $t_1$, pair 3 uses $D_1$ at $t_2$, leaving pair 1 only one candidate ($D_2$) to choose. However, $D_2$ was taken by pair 2 at $t_0$ which is unknown to pair 1.}
\label{fig:2channel}
\end{figure}

Now see \fref{fig:exp_ener} for power consumption. {\sl Altruistic} consumes the lowest power among all the protocols when there are a sufficient number of nodes. Another observation is, although experiments and simulations both use the same power consumption rates, the experiment statistics are consistently lower than the simulation statistics (\fref{fig:sg_ener_highload}). This is because we prolonged the protocol intervals in our hardware implementation to overcome the inaccurate timing of TelosB, as described in \sref{sec:impl}, and hence the IDLE state appears more often and the TX/RX state appears less often in experiments than in simulations.

In summary, the testbed experiments confirm that {\sl Altruistic} achieves high throughput and low energy consumption simultaneously, and is the most cost-efficient among all the protocols under comparison. Our work also shows that multi-channel MAC protocols can be indeed implemented on COTS hardware and work with a single radio and asynchronously.

\fi

\section{Discussion}\label{sec:discuss}

\subsection{Limitations}\label{sec:limit}

Altruistic DISH becomes less effective when there are only a few peers (compared to the number of channels) or traffic is light, in which case channel contention is very mild. For instance, {\sl Altruistic} achives lower BMP than {\sl Non-DISH-psm} in \fref{fig:bmp_sg} at 10 nodes (5 data channels) under low traffic, and similarly in \fref{fig:exp_bmp} at 2 nodes. In such scenarios, in-situ energy conscious DISH could be a better choice as it is able to reduce cooperation by adapting to network dynamics.

Another limitation is that the four-way control channel handshake in the DISH protocols can incur more overhead than usual protocols. Although this can be largely offset by the cooperation gain, it is still desired to reduce the overhead. One effective way is to use {\em packet train} to amortize the overhead, which was also used by MMAC\cite{mmac04}, SSCH\cite{ssch04}, and WiFlex\cite{wiflex08}. We have adopted this technique in \cite{tie09secon} for cognitive radio networks.

\subsection{Alternative Methods for Altruistic DISH}\label{sec:alternative}

An alternative method for altruistic DISH is to add one more radio on a few peers and let these additional radios act as altruists. This may further enhance the cost efficiency as the cost of a radio is much lower than the cost of a node. The trade-off is the need of designing a multi-radio MAC protocol which, particularly, must coordinate the use of the control channel shared by the two co-located radios. As the hardware platform (TelosB) does not support multiple radios, this alternative method merits our future study that adopts a different platform.

Another alternative to prolong network lifetime is to add an extra battery to each existing node instead of adding altruists. This is simple but would present a challenge to the size of each node, be it a laptop, a mobile, or a PDA. Also, from the perspective of scalability, the additional cost (due to extra batteries) will increase linearly when the number of peers increases, whereas in the altruistic approach, the additional cost (due to extra nodes, i.e., altruists) remains constant (as shown in  \sref{sec:deploy}).

\subsection{Energy Fairness}\label{sec:fair}

A possible concern is that, being always awake, altruists may be over-burdened and drain energy very fast. A possible solution is to apply the in-situ strategy on top of altruist DISH such that altruists rotate the role of cooperation. However, this will sacrifice simplicity which is a primary advantage of the altruist strategy. Furthermore, having altruists stay awake is not necessarily energy unfair because our evaluation in terms of BMP, which already takes energy fairness into account (via $P_p^{max}$ and $P_a^{max}$, see \eqref{eq:bmp}), has shown (in both simulation and testbed) that altruistic DISH performs very well in most cases. Nonetheless, fairness might be a problem under non-uniform traffic patterns and thus merit future study.

\ifdefined\thesis
\else
\section{Related Work}\label{sec:relwork}

\ifdefined\thesis
\section{Energy-Efficient Multi-Channel MAC Protocols}\label{sec:eemc}
\else
\subsection{Energy-Efficient Multi-Channel MAC Protocols}\label{sec:eemc}
\fi

There are a few proposals on this new topic. In ad hoc networks, PSM-MMAC~\cite{wang06infocom} lets nodes to choose to be awake or doze based on the estimated number of active links, queue length and channel condition. TMMAC~\cite{tmmac07} uses the 802.11 ATIM window like MMAC~\cite{mmac04}, but in addition to negotiating channels, it also negotiates time slots for nodes to sleep in.

In wireless sensor networks (WSNs), MMSN~\cite{mmsn06} was proposed to use multiple channels. However, energy saving is not one of its design goals, but is a natural and common consequence of using multiple channels (as interference is reduced). Also, when the number of channels is small, it can be seen from the paper that MMSN consumes more energy than single-channel CSMA.
\cite{cit06wsn} proposes another protocol for cluster-based WSN. The protocol is shown to be more energy efficient than MMSN by assuming (i) all cluster heads can {\em directly} communicate with each other and (ii) there are many sink nodes and hence no single-sink bottleneck. The practicality of these assumptions can be questioned. CMAC~\cite{cmac06wcnc}, unlike MMSN and \cite{cit06wsn} which are both synchronous protocols, does not require time synchronization. However, it needs to assign every node a channel that does not overlap with any other node in 2-hop range. This means that for a network with a node density of, say, 10/$r^2$, at least 126 channels are needed, which is generally not feasible.

Our work \ifdefined\thesis described in \cref{chap:energy} \fi differs from existing work in the following: (1) instead of proposing a {\em protocol}, we propose {\em strategies} which can generally apply to a class of protocols (DISH-based protocols), (2) we do not require multiple radios as in PSM-MMAC and CMAC, nor time synchronization as in TMMAC, MMSN and \cite{cit06wsn}, and (3) our proposal can be used in both single-hop and multi-hop networks, unlike PSM-MMAC which supports WLAN only.

\ifdefined\thesis
\section{Energy-Efficient Single-Channel MAC Protocols}\label{sec:sleepwake}
\else
\subsection{Energy-Efficient Single-Channel MAC Protocols}\label{sec:sleepwake}
\fi

In ad hoc networks, Tseng et al.~\cite{tseng02infocom} proposed three power-saving protocols for multi-hop scenarios, with time synchronization not required. These protocols differ in their power saving capability and neighbor discovery time, and can be chosen according to specific application needs.

In WSNs, there are lots of proposals and most of them can be applied to or adapted for static ad hoc networks as sensor devices are more resource-constrained.
In S-MAC~\cite{smac04ton}, nodes in each neighborhood negotiate a sleep-wake schedule in order to wake up at the same time. Nodes on the border of two adjacent neighborhoods will maintain two schedules to keep connectivity. In this way, network-wide synchronization is not required. T-MAC\cite{tmac03sensys} improves S-MAC by shortening the awake period when there is no communication request. Each node wakes up at the start of an awake period, listens to the channel for a short time and, if there is no incoming data, returns to sleep immediately without waiting for the end of the awake period.
B-MAC~\cite{bmac04sensys} introduces low power listening and long preamble transmission: each data packet has a preamble slightly longer than a node's sleep period, and hence a receiver is always able to detect the transmission from a sender. Time synchronization is not required.
X-MAC~\cite{xmac06sensys} improves B-MAC by embedding a receiver address into the preamble and strobing the preamble, so that nodes who are not the intended receiver can return to sleep earlier.

\ifdefined\thesis
\section{Multi-Channel MAC Testbed}\label{sec:impl-relwork}
\else
\subsection{Multi-Channel MAC Testbed}\label{sec:impl-relwork}
\fi

There are a few hardware implementations of multi-channel MAC protocols. Chereddi et al.~\cite{realman06hybrid} reported a 4-node single-hop network testbed implemented on Linux with Atheros chipset, for a hybrid multi-channel MAC protocol proposed in \cite{kya06mc2r}. The protocol is based on a channel abstraction module and requires two interfaces per node: one is tuned to a fixed channel for packet receiving and the other switches channels for packet transmission.
McMAC~\cite{mcmac07wcnc} uses a single radio and was implemented on Telos~\cite{telos} as a proof of concept. However, the implementation was a simplified version which does not measure performance metrics such as throughput, delay, or energy consumption (the only reported performance was how long it takes to synchronize sender-receiver pairs onto common channels). Y-MAC~\cite{ymac08ipsn} is another single-radio multi-channel MAC but is proposed for WSNs. It is TDMA based and specifically deals with bursty traffic in dense WSNs. It classifies every time slot as a send or a receive slot, and divides each slot into a contention window and a send/receive window.
The protocol was implemented in RETOS~\cite{retos07ipsn} on TmoteSky motes~\cite{tmotesky}, and demonstrated low duty cycle and low delivery latency via experiments. However, throughput was not measured.

All the above protocols require time synchronization (\cite{kya06mc2r} needs loose synchronization). Recently, So et al. \cite{so06sync} showed that it is difficult to achieve synchronization in multi-channel networks and it incurs significant overhead. They also implemented a multi-channel time-synchronizing protocol, but the protocol only exchanges beacons and does not handle data packets (see Section 7.1 therein).

Most recently, there appeared two implementations of {\em asynchronous} multi-channel MAC protocols, both for WSNs. One is TMCP \cite{infocom08tmcp}, designed for {\em data collection} applications (the traffic considered was many-to-one CBR streams) and for networks with only a small number of channels. A network is partitioned into multiple subtrees and each subtree is allocated a different channel. The authors implemented the protocol on MicaZ motes and evaluated packet delivery ratio, which reflects throughput to some extent, but energy was not evaluated. Le et al.~\cite{ipsn08prac} built a multi-channel MAC testbed also using MicaZ motes and evaluated performance in terms of the number of received messages. The energy issue was not specifically considered. Like TMCP, the protocol was designed for WSN data collection and aggregation applications. Under the random traffic pattern, which is typical in ad hoc networks, it will lead to poor performance (see Section 6 therein).

\ifdefined\thesis Our testbed, as will be addressed in \cref{chap:impl}, \else Our testbed \fi differs from prior work in that (1) it is designed for ad hoc networks using a single radio per node and not using time synchronization, (2) it is able to evaluate typical performance metrics such as throughput and energy consumption, and (3) it is a full implementation of all the protocol functionalities.

\fi

\ifdefined\thesis
\section{Summary}\label{sec:conclud}
DISH %
\else
\section{Conclusion}\label{sec:conclud}
Distributed information sharing (DISH) %
\fi
can significantly boost the system throughput for multi-channel MAC protocols, but it also heighten the energy consumption due to its information sharing component (which subsumes information gathering as well). In this paper, we propose two energy-efficient strategies and conduct a comparative study on five protocols that differ in the usage of DISH and the strategies. Both simulations and testbed experiments show that altruistic DISH (1) is a very simple strategy which does not involve protocol re-design or incur additional runtime overhead, (2) substantially reduces energy consumption while maintaining (sometimes even enhancing) the throughput benefit from DISH, and also (3) apparently improves cost efficiency. The other strategy, in-situ energy conscious DISH, is suitable for applications with few nodes or light traffic, or those that preclude using additional nodes.

The key to the success of altruistic DISH is twofold. First, using altruists as dedicated cooperative nodes provides cooperation in a {\em guaranteed}, as opposed to opportunistic, manner. Second, the use of altruists shifts the resource-consuming tasks (information gathering and sharing) from {\em all} nodes to only {\em a few}.

Altruistic DISH clearly separates the data plane and the control plane: peers are solely responsible for forwarding data traffic and altruists are solely responsible for control-plane cooperation, i.e., DISH.

This \ifdefined\thesis study \else paper \fi gives the first treatment on energy efficiency for cooperative multi-channel MAC protocols. We believe that DISH is an approach worth exploring and that altruistic DISH is a simple yet effective strategy to implement DISH.

\end{document}